\documentclass{article}[11pt]

\usepackage{latexsym}
\usepackage{amssymb}

\usepackage{epsfig}

\usepackage{psfrag}
\usepackage[all]{xy}
\usepackage[active]{srcltx}

\usepackage{floatpag}
\usepackage{chngpage}
\usepackage{latexsym}
\usepackage{amsfonts}
\usepackage{epsfig}
\usepackage{amsmath}
\usepackage{amscd}
\usepackage{amssymb, amsmath}
\usepackage[a4paper]{geometry}
\usepackage{amssymb,latexsym,amscd}
\usepackage{epsfig}
\usepackage[active]{srcltx}
\usepackage{graphics,graphicx}
\usepackage{pstricks,pst-node,pst-tree}
\usepackage{braket}
\usepackage{color}
\usepackage{bm}
\usepackage[all]{xy}
\usepackage{lscape}
\usepackage{psfrag}
\usepackage{array}
\usepackage{rotating}
\usepackage{lipsum}
\usepackage{amssymb,  amsmath}
\usepackage{amssymb, amsthm, amsmath}
\usepackage[a4paper]{geometry}
\usepackage{amssymb,latexsym,amscd}
\usepackage{epsfig}
\usepackage[active]{srcltx}
\usepackage{graphics,graphicx}
\usepackage{pstricks,pst-node,pst-tree}
\usepackage{braket}
\usepackage{color}
\usepackage{bm}
\usepackage[all]{xy}
\usepackage{lscape}
\usepackage{psfrag}
\usepackage{array}
\usepackage{rotating}

\definecolor{r}{rgb}{1,0,0}
\definecolor{b}{rgb}{0,0,1}
\definecolor{g}{rgb}{0,1,0}

\renewcommand{\P}{L}
\newcommand{\B}{\mathrm{B}}
\newcommand{\Q}[1]{\mathrm{Ch}_{#1}}
\newcommand{\id}{\mathtt{id}}
\newcommand{\sgn}{\mathtt{sgn}}
\newcommand{\CC}{\mathbb{C}}
\newcommand{\LL}{\mathfrak{L}}
\newcommand{\SG}{\mathfrak{S}}
\newcommand{\G}{\mathcal{G}}
\newcommand{\HH}{\mathcal{H}}

\newcommand{\ZZ}{\mathbf{D}_4}
\newcommand{\ZZo}{\mathfrak{S}_2 \wr \mathfrak{S}_2 }
\newcommand{\dimC}{ \dim_{\mathbb{C}}}
\newcommand{\dimR}{ \dim_{\mathbb{R}}}
\newcommand{\idmat}{\mathbf{1}}
\newcommand{\fra}[2]{\textstyle{\frac{#1}{#2}}}

\newcommand{\lrc}[1]{{\langle #1 \rangle}_{\mathbb{C}}}

\newcommand{\lie}[2]{\left[#1,#2\right]}
\renewcommand{\1}{\mathbf{1}}

\newcommand{\trip}[3]{(#1,#2,#3)}
\renewcommand{\qed}{\hfill $\diamond$}

\theoremstyle{plain}
\newtheorem{theorem}{Theorem}[section]

\theoremstyle{definition}
\newtheorem{definition}[theorem]{Definition}

\newtheorem{example}{Example}
\newtheorem{remark}{Remark}
\newtheorem{note}{Note}
\renewenvironment{proof}{\noindent {\textsc{Proof.}}}{$\square$ \vspace{3mm}}

\newif\ifprivate
\privatetrue

\def\???{\ifprivate {\bf {???}} \marginpar{{\Huge {\bf ?}}}
\else \fi}

\title{Lie Markov models with purine/pyrimidine symmetry}
\author{ Jes\'us Fern\'andez-S\'anchez$^{1}$, Jeremy G. Sumner$^{2}$,\\ Peter D. Jarvis$^{3}$, Michael D.  Woodhams$^{4}$}

\begin{document}
\maketitle

\begin{abstract}
Continuous-time Markov chains are a standard tool in phylogenetic inference. 
If homogeneity is assumed, the chain is formulated by specifying time-independent rates of substitutions between states in the chain. 
In applications, there are usually extra constraints on the rates, depending on the situation.
If a model is formulated in this way, it is possible to generalise it and allow for an inhomogeneous process, with time-dependent rates satisfying the same constraints. 
It is then useful to require that there exists a homogeneous average of this inhomogeneous process within the same model.
This leads to the definition of ``Lie Markov models'', which are precisely the class of models where such an average exists.
These models form Lie algebras and hence concepts from Lie group theory are central to their derivation. 
In this paper, we concentrate on applications to phylogenetics and nucleotide evolution, and derive the complete hierarchy of Lie Markov models that respect the grouping of nucleotides into purines and pyrimidines -- that is, models with purine/pyrimidine symmetry. 
We also discuss how to handle the subtleties of applying Lie group methods, most naturally defined over the complex field, to the stochastic case of a Markov process, where parameter values are restricted to be real and positive.
In particular, we explore the geometric embedding of the cone of stochastic rate matrices within the ambient space of the associated complex Lie algebra.

The whole list of Lie Markov models with purine/pyrimidine symmetry is available at \verb+http://www.pagines.ma1.upc.edu/~jfernandez/LMNR.pdf+.
\end{abstract}

\vspace{3mm}
\noindent
\textit{keywords:} phylogenetics, Lie algebras, Lie groups, representation theory, symmetry, Markov chains

\vfill
\hrule\mbox{}\\
\thanks{\footnotesize{
\noindent
$^1$Departament de Matem\`atica Aplicada I, Universitat Polit\`ecnica de Catalunya, Spain \\ e-mail: jesus.fernandez.sanchez@upc.edu\\
$^{234}$ School of Mathematics and Physics, University of Tasmania, Australia\\ $·^2$ ARC Research Fellow; e-mail: jsumner@utas.edu.au\\ 
$^3$ Alexander von Humboldt Fellow; e-mail: peter.jarvis@utas.edu.au\\$^4$ e-mail: michael.woodhams@utas.edu.au\\

}
}

\section{Introduction}
\label{intro}

Most of the commonly implemented models in molecular phylogenetics are based on the continuous-time Markov assumption. 
For these models, molecular substitution events (along an edge of a phylogenetic tree) are ruled by substitution rates. 
For DNA models -- where the state space consists of the four nucleotides adenine, cytosine, guanine and thymine -- twelve substitution rates must be specified for each edge of the evolutionary tree, and the precise characteristics of the process are fixed by constraints on these rates. 
These constraints define a space of parameter values where each point corresponds to unknown evolutionary quantities such as base composition and mutation rate. 
As in all applied statistics, there is a trade-off between more complex, realistic models, and simpler, tractable models:  complex models can provide very close fits to the observed data, but are more vulnerable to random error. 
A standard assumption in molecular phylogenetics is to work with homogeneous Markov chains, where the substitution rates are assumed to be constant in time.  
 
The motivation behind our previous work \cite{LMM} was to consider the consequences of allowing for  some change in individual substitution rates which may well have occurred independently across the evolutionary history. 
With this perspective, the evolutionary process can still be modeled as an continuous-time Markov chain, but we must allow the process to be inhomogeneous, where the rates are allowed to vary as a function of time throughout the evolutionary history. 
This leads to considering evolutionary model classes that are ``multiplicatively closed'', that is, models where the product of substitution matrices is still in the model. 
For such models, it it is possible to interpret the time average of their inhomogeneous behaviour as a homogeneous process within the same model class. 
Many oft-used models, such as the general time-reversible model \cite{tavare1986,posada1998}, are not multiplicatively closed and this deficiency poses a problem for phylogenetic analysis in both flexibility of interpretation, and as a potential source of model-misspecification \cite{sumner2012}. 
For a multiplicatively closed model, it is possible to model evolutionary processes homogeneously, by interpreting the fitted substitution rates as an ``average'' of the true inhomogeneous process occurring on each branch of the tree. 
In \cite{LMM}, we presented sufficient conditions for multiplicative closure of continuous-time Markov chains,  and this lead directly to the concept of Lie Markov models.
These models arise when we demand  the set of rate matrices of the model  form a Lie algebra. 
This is a technical condition guaranteeing the corresponding set of substitution probability matrices will be multiplicatively closed, as desired. 
Mathematically, Lie Markov models can be regarded as a generalisation of other model classes, such as equivariant \cite{draisma2008} or group-based  \cite[Chap. 8]{semple2003} models. 

In \cite{LMM} we discussed the symmetry properties of DNA models to nucleotide permutations, and noted the statistical relevance of these symmetries to likelihood calculations.
The main result of that paper was a procedure to generate multiplicatively closed Markov models with a prescribed symmetry, which has desirable properties in terms of model selection. 
For instance, a biologist may wish that candidate models do not provide any natural groupings of nucleotides, and hence $\SG_4$ symmetry -- i.e. the symmetry group of all possible nucleotide permutations -- is appropriate. It is then a matter of choosing how many free parameters are appropriate for the given data set. 
The complete hierarchy of Lie Markov models with $\SG_4$ symmetry was derived in \cite{LMM}.

\begin{table}
\small{
\begin{tabular}{|c|c|m{8cm}|}
\hline
Model & Rate Matrix & Description \\
\hline \hline
%
\mbox{3.3b} & 
$\left(\begin{array}{cccc}
* & \alpha & \beta & \gamma \\
\alpha &* & \gamma & \beta \\
\gamma & \beta &* & \alpha \\
\beta & \gamma & \alpha & *
\end{array}\right)$  & A 3-dimensional model, with parameter $\alpha$ for transitions $A\leftrightarrow G, C\leftrightarrow T$, and two different parameters for transversions: $\beta$ for $A\mapsto C \mapsto G \mapsto T \mapsto A$ and $\gamma$ for $A\mapsto T \mapsto G \mapsto C \mapsto A$.\\ 
\mbox{3.3c} & 
$\left(\begin{array}{cccc}
*& \alpha & \beta & \beta \\
\alpha &* & \beta & \beta \\
\beta & \beta & * & \gamma \\
\beta & \beta & \gamma &*
\end{array}\right)$ &  A 3-dimensional model, with two parameters for transitions: $\alpha$ for $A\leftrightarrow G$ and $\gamma$ for  $C\leftrightarrow T$, and one parameter for transversions: $\beta$ for $A \leftrightarrow C \leftrightarrow G \leftrightarrow T \leftrightarrow A$. \\
& & Note the similarity of both of these models to the Kimura 3 parameter model \cite{kimura1981}, which also belongs to our hierarchy as Model 3.3a. \\
\hline
4.4a & 
$\left(\begin{array}{cccc}
* & \alpha  & \alpha & \alpha \\
\beta & *& \beta & \beta \\
 \gamma & \gamma & * & \gamma \\
\delta & \delta & \delta & *
\end{array}\right)$
 & The Felsenstein 1981 model \cite{felsenstein1981}. A 4-dimensional model, where mutations share the same rate when they change to the same nucleotide.  \\
\mbox{4.4b} & 
$\left(\begin{array}{cccc}
* & \alpha  & \beta & \beta \\
\alpha & *& \beta & \beta \\
\gamma & \gamma & * & \delta \\
\gamma & \gamma & \delta & *
\end{array}\right)$ & A 4-dimensional model, with two parameters for transitions: $\alpha$ for $A\leftrightarrow G$ and $\gamma$ for  $C\leftrightarrow T$, and the other two parameters for transversions: $\beta$ for $A\mapsto C \mapsto G \mapsto T \mapsto A$ and $\gamma$ for $A\mapsto T \mapsto G \mapsto C \mapsto A$.\\ 
\hline 
\mbox{5.6b} & 
$ \left(\begin{array}{cccc}
* & a+x & b+ x & b + x \\
a +y & * &b +y & b +y \\
b + z & b + z &* & a + z \\
b +t & b +t & a +t & *
\end{array}\right)$ & A 5-dimensional model, where rates depend on two families of parameters $\{a,b\}$ and $\{x,y,z,t\}$: transitions and tranversions have parameters $a$ and $b$ respectively, while they are affected by some other parameters according to the nucleotide they change to: $x$ for mutations to $A$, $y$ for mutations to $G$, $z$ for mutations to $C$ and $t$ for mutations to $T$. These parameters are subjected to the constraint $x+y+z+t=0$. Notice the resemblance of this model with HKY85 (see Remark \ref{HKY}). \\ \hline
\mbox{6.6} & 
$\left(\begin{array}{cccc}
* & \alpha & \beta & \gamma \\
\alpha & *& \gamma & \beta \\
\delta & \varepsilon & * & \zeta\\
\varepsilon & \delta & \zeta & *
\end{array}\right)$ & This 6-dimensional model has two parameters for different transitions: $\alpha$ for $A\leftrightarrow  G$ and $\zeta$ for  $C \leftrightarrow  T$, and 4 for transversions: $\beta: A \mapsto C, G\mapsto T$, $\gamma: A \mapsto T, G\mapsto C$, $\delta: C \mapsto A, G\mapsto T$, $\varepsilon: A \mapsto T, G\mapsto C$. By permuting rows and columns acoording to $(GT)$ (or $(AC)$), we obtain the \emph{strand symmetric} model (see \cite{cassull}). 
\\ \hline
\mbox{6.7a} & 
$\left(\begin{array}{cccc}
* & a + x & b + x & c + x \\
a + y & *& c + y & b + y \\
b +z & c+z & * & a +z \\
c +t & b +t & a +t & *
\end{array}\right)$ & This 6-dimensional model is similar to 5.6b, but  it has two parameters for different transversions: $b$ for  $A\mapsto C \mapsto G \mapsto T \mapsto A$ and $c$ for $A\mapsto T \mapsto G \mapsto C \mapsto A$.\\ \hline
\end{tabular}
}
\caption{\small{\label{some_models} Some Lie Markov models with purine/pyrimidine symmetry that may have special interest for biologists.}}
\end{table}

In this paper, we deal with the case of closed Markov models whose symmetry is consistent with the grouping of nucleotides into purines and pyrimidines, i.e. $AG\mid CT=\{\{A,G\},\{C,T\}\}$. 
As will be discussed, this motivates us to produce and examine the Lie Markov models with  symmetry governed by the permutation subgroup of $\SG_4$ that preserves the purine/pyrimidine grouping\footnote{Note this group is isomorphic to the dihedral group $\ZZ$, which describes the symmetries of a square.
However, it also admits a more natural description in our setting as $\ZZo$, the wreath product of $\mathfrak{S}_2$ with itself (see Chapter VII of \cite{rotman} for instance).} :
\[\G:=\{e,(AG),(CT),(AG)(CT),(AC)(GT),(AT)(CG),(ACGT),(ATGC)\},\]
where $e$ is the identity, or ``do nothing'', permutation.

We will also go further than \cite{LMM} by exploring the definition of these models and investigate the geometrical properties  that arise naturally when we deal with the tension between the algebraic formalism of Lie groups, where one works over the complex field, and the stochastic constraints of Markov models, where parameter values are constrained to be real and positive. 
In particular, we discuss the geometric embedding of the stochastic rate matrices within the vector space of complex rate matrices. 
These considerations motivate our definition of the \emph{stochastic cone} of a Lie Markov model. 
Besides its geometrical interest, the stochastic cone is the set of stochastic rate matrices of the model and in a practical context is actually the main object of interest. 
We will discuss implementation and performance of the models we present here in a sister paper \cite{woodhams2012}.

Although our presentation focuses on the purine/pyrimidine grouping $AG \mid CT$, given the appropriate nucleotide permutation, exactly the same hierarchy of models would arise if we were to consider the grouping $AC \mid GT$, or the grouping $AT \mid GC$.
The reader should note that choosing the grouping $AT \mid GC$ would give the classification of all Lie Markov models that preserve \emph{complementation} $A\leftrightarrow T$, $C\leftrightarrow G$ (see \cite{yap_pachter}). In particular, the ``strand-symmetric'' model defined by   Casanellas and Sullivant in  \cite{cassull} arises in this way from our Model 6.6 (see Table \ref{some_models}). The conversion of our hierarchy of models from the $AG \mid CT$ grouping to the $AT \mid GC$ grouping would follow by simultaneously permuting  the $G$ and $T$ rows and columns of the rate matrices in each model.

In Section~\ref{sec2} we recall some of the basic definitions and tools introduced in \cite{LMM}. 
We revisit the definition of Lie Markov models, and introduce the concept of the stochastic cone of a Lie Markov model. 
We also recall the basic results on group theory and representation theory  necessary for the development of our results. 
In Section~\ref{sec3} we recall the idea of Lie Markov model with prescribed symmetry given by a permutation group $G$. 
We introduce the ray-orbits of the corresponding stochastic cone, which are the orbits under the action of $G$ of the rays of the stochastic cone.  
In Section~\ref{sec4}, we take $G=\G $, decompose the space of rate matrices as a $\G$-module and provide a basis consistent with this decomposition. 
We also determine the isomorphism classes of possible $\G$-orbits and the decomposition of their (abstract) span into irreducible modules. 
In Section~\ref{sec5}, we give the whole list of Lie Markov models with purine/pyrimidine  symmetry. 
Each model is given by exhibiting a basis of the corresponding space of matrices as well as the ray-orbits of its stochastic cone. 
Among these models, we obtain the Jukes-Cantor model, the Kimura models with 2 or 3 parameters,
the general Markov model and a number of new models that may have special interest for the biologists. Some of these are shown in Table \ref{some_models} as an appetizer before the whole list, which itself can be found in explicit form at 
\verb+http://www.pagines.ma1.upc.edu/~jfernandez/LMNR.pdf+.
Finally, in the conclusions we discuss implications and possibilities for future research.

\section{Preliminaries}\label{sec2}
Throughout this section, we will recall some definitions and basic facts from \cite{LMM}, which we also refer to for some proofs. 
We keep the assumptions and the notation already introduced there. 
In particular, we work over the complex field $\mathbb{C}$, and for simplicity refer to a matrix as ``Markov'' if the entries in each column sum to one. 
Later we will discuss how to specialise to the stochastic case where the entries must be real numbers in the range $\left[0,1\right]$.
This will lead us to considering the stochastic cone of the Lie Markov model, which will be the set of real rate matrices with non-negative entries outside the diagonal. 

We define the \emph{general Markov} model $\mathfrak{M}_{GM}$ as the set of $n\times n$ matrices whose columns sum to one: 
\[\mathfrak{M}_{GM}:=\left\{M\in \mathbb{M}_n(\mathbb{C}) : \bm{\theta}^T M =\bm{\theta}^T \right\},\]
where $\bm{\theta}$  is the column $n$-vector with all its entries equal to $1$, i.e. $\bm{\theta}^T=(1,1,\ldots,1)$.
Recall that, in a  homogeneous continuous-time Markov chain, the corresponding Markov matrices occur as exponentials $M=e^{Qt}$, where $Q$ is a ``rate matrix'' and $t$ is time elapsed.
We write $\mathfrak{L}_{GM}=\left\{ Q\in M_n(\mathbb{C}) : \bm{\theta}^T Q =\bm{0}^T \right\}$, to indicate the set of all (complex) rate matrices.
We refer to a Markov matrix $M\in \mathfrak{M}_{GM}$, or a rate matrix $Q\in \mathfrak{L}_{GM}$, as ``stochastic'' if its off-diagonal elements are real and positive.

Under matrix multiplication, the set
\begin{eqnarray*}
GL_1(n,\mathbb{C}):=\left\{M\in \mathbb{M}_n(\mathbb{C}) : \bm{\theta}^T M =\bm{\theta}^T ,\det(M)\!\neq \!0\right\},
\end{eqnarray*}
forms a subgroup of the general linear group of invertible $n\times n$ matrices with complex entries, i.e. $GL_1(n,\mathbb{C})< GL(n,\mathbb{C})$.
It contains the matrix exponential of any rate matrix, that is, 
\begin{eqnarray*}
 e^{\mathfrak{L}_{GM}}:=\left\{e^Q : Q\in \mathfrak{L}_{GM}\right\}
\subset  GL_1(n, \mathbb{C}). 
 \end{eqnarray*}
We refer to $e^{\mathfrak{L}_{GM}}$ as the \emph{general rate matrix model}.

A \emph{Markov model} $\mathfrak{M}$ is some subset  $\mathfrak{M}\subseteq \mathfrak{M}_{GM}$ of the general Markov model containing the identity matrix $\mathbf{1}$.
A Markov model $\mathfrak{M}$ is \emph{multiplicatively closed} if for all $M_1$ and $M_2 \in \mathfrak{M}$ we also have the matrix product $M_1M_2\in \mathfrak{M}$.
Similarly, given a subset $\mathfrak{L}\subseteq \mathfrak{L}_{GM}$ of rate matrices, we refer to $e^{\mathfrak{L}}$ as a \emph{rate matrix model}.  
It is clear that all rate matrix models are Markov models, and we simplify terminology and also refer to $\mathfrak{L}$ as a ``model''.

We are primarily interested in rate matrix models $\mathfrak{M}=e^{\LL}$ which are multiplicatively closed. 
For such a model $\mathfrak{M}$, suppose it  is a smooth manifold  around the identity matrix $\idmat$, so there exist differentiable paths $A(t)\in \mathfrak{M}$ with $A(0)=\idmat$.
Then we can define the tangent space at the identity: $T_{\1}(\mathfrak{M})=\{A'(0):A(t)\in \mathfrak{M}, A(0)=\idmat \}$. 
Then, because of the Baker-Campbell-Hausdorff formula \cite{campbell1897}, if $T_{\1}(\mathfrak{M})$ forms a Lie algebra, $\mathfrak{M}$ is multiplicatively closed in a neighbourhood of $\mathbf{1}$.

Recall  $T_{\1}(\mathfrak{M})$ is a Lie algebra if for all $Q_1, Q
_2\in T_{\1}(\mathfrak{M})$ and $\lambda \in \mathbb{C}$:
\begin{enumerate}
\item $Q_1+\lambda Q_2\in T_{\1}(\mathfrak{M})$,
\item $\lie{Q_1}{Q_2}:=Q_1Q_2-Q_2Q_1\in T_{\1}(\mathfrak{M})$.
\end{enumerate}
The first condition states that $T_{\1}(\mathfrak{M})$ is a vector space, and the second states that $T_{\1}(\mathfrak{M})$ is closed under ``Lie brackets''.
%

Presently, we recall from \cite{johnson1985} and \cite{LMM} the Lie algebra structure of the general Markov model. 
To this aim, consider the set of ``elementary'' rate matrices $\{L_{ij}:1 \leq i\neq j\leq n\}$, where $L_{ij}$ is the $n\times n$ matrix with 1 in the $ij$ entry, -1 in the $jj$ entry and 0 everywhere else.
The matrices $\{L_{ij}\}_{i\neq j}$ form a $\mathbb{C}$-basis for the tangent space of $GL_{1}(n,\mathbb{C})$
and, in particular, we can express any rate matrix $Q$ as a linear combination:
\begin{eqnarray}\label{exp_rates}
Q=\sum_{i\neq j}\alpha_{ij}L_{ij}. 
\end{eqnarray}
This is a convenient basis for $\mathfrak{L}_{GM}$ because the stochastic condition on $Q$ is simply that the coefficients $\alpha_{ij}$ are real and non-negative. 
Moreover, if $\delta_{ij}$ denotes the Kronecker delta ($\delta_{ii}=1$ and $\delta_{ij}=0$ when $i\neq j$), the equalities 
\begin{eqnarray*}
 [L_{ij},L_{kl}]=(L_{ij}-L_{jl})(\delta_{jk}-\delta_{jl})-(L
_{kj}-L_{lj})(\delta_{il}-\delta_{jl})
\end{eqnarray*}
exhibit the Lie algebra structure of $\LL_{GM}$.

Given a vector subspace $\mathfrak{L}\subset \mathfrak{L}_{GM}$, a \emph{stochastic generating set} for $\mathfrak{L}$ is a generating set $B_\LL
=\{L_1,L_2,\ldots ,L_d\}$ of $\mathfrak{L}$ such that each $L_k$ is a non-negative linear combination of the $L_{ij}$, i.e.  $L_{k}=\sum_{i\neq j}\alpha_{ij}L_{ij}$ where $\alpha_{ij}\geq 0$. 
A \emph{stochastic basis} of $\LL$ is a stochastic generating set where the vectors are linearly independent.

\begin{definition}[cf. \cite{LMM}]\label{def:LMM}
A \emph{Lie Markov model} is a Lie subalgebra $\LL$ of $\LL_{GM}$ for which there exists a stochastic basis. 
\end{definition}

Leaving the technical aspects aside, a Lie Markov model is a model for which the product of two substitution matrices is still in the model. The motivation of such models is given by the fact a non-homogeneous evolutionary processes can be described in a homogeneous fashion. In more concrete terms, if $M_1, M_2 \in e^{\LL}$, then $M_1 M_2 \in e^{\LL}$,  i.e. for any inhomogeneous process on an edge where the rate matrices always lie within $\LL$, there is an equivalent homogeneous process on that edge, whose rate matrix also lies within $\LL$. This is not the case for the general time reversible model (the reader is referred to \cite{LMM} for a detailed proof of the non-closure of GTR). 

\begin{remark}
By an elementary result in linear algebra, any generating set for a vector space can be reduced to a basis by removing elements, and hence Definition~\ref{def:LMM} would remain unchanged if ``stochastic basis'' were replaced with ``stochastic generating set''. \qed
\end{remark}

\begin{figure}[t]
\begin{center}
\includegraphics[scale=0.5]{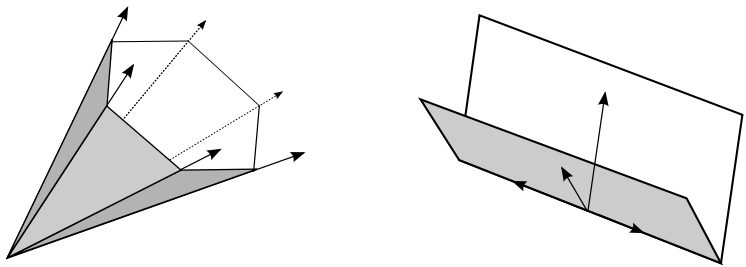}
\caption{\label{fig:sto_cone} \small{A strongly convex polyhedral cone of dimension 3 with 6 rays (represented by arrows) and a convex polyhedral cone which is not strongly convex.}}
\end{center}
\end{figure}

We are especially interested in the study of the set of stochastic rate matrices of the model. 
The condition of Definition~\ref{def:LMM} ensures  $\LL$ contains \emph{enough} stochastic rate matrices (see Theorem~\ref{res:stochastic_cone} forthcoming), and it is useful to give some geometrical interpretation of this condition. 
To this aim, we need to recall some basic definitions on convex polyhedral cones. 

Following \cite{convex}, a \emph{convex polyhedral cone} in $\mathbb{R}^n$ is defined as a set 
\begin{eqnarray*}
 C=\{\lambda_1 v_1+\ldots +\lambda_r v_r : \lambda_i\geq 0\}
\end{eqnarray*}
generated by some finite set of vectors $v_1,\ldots,v_r$ in $\mathbb{R}^n$. Such vectors are called \emph{generators} of the cone $C$.
The reader may note that,  with this definition, every linear subspace of $\mathbb{R}^n$ is a convex polyhedral cone.  When a convex polyhedral cone contains no nonzero linear subspaces, it is said to be \emph{strongly convex}. 
In this case, which has special interest for us, any minimal system of generators of the cone is unique up to multiplication with positive scalars 
\cite{convex}.
The \emph{rays} of the cone are the non-negative spans of each vector in a minimal system of generators, and they correspond to the 1-dimensional faces of the cone \cite{convex}; see Figure~\ref{fig:sto_cone} for an illustration. 
%
%
Farkas's theorem ensures the polyhedral cones can be equivalently defined as the intersection of finitely many halfspaces. 
It follows that the intersection of any two convex polyhedral cones in $\mathbb{R}^n$ is again a convex polyhedral cone. 
\begin{note}
Consider a collection of vectors $X=\{X_1,\ldots, X_r \}$.
In what follows we will use the notation $\mathbb{F}X$ or $\langle X_1,\ldots , X_r \rangle_\mathbb{F}$ to indicate the linear span of $X$ over the field $\mathbb{F}$, where $\mathbb{F}=\mathbb{R}$ or $\mathbb{C}$.
That is, \[\mathbb{F}X= \langle X_1,\ldots, X_r \rangle_\mathbb{F} :=\{\lambda_1 X_1+\ldots +\lambda_r X_r:\lambda_i \in \mathbb{F}\}.\]
Of course, $\mathbb{F}X$ is a vector space. In particular, we can consider $V:=\mathbb{C}X$ as a complex vector space with dimension $r$, or as a real vector space $V=\mathbb{R}X+\mathbb{R}({\bf i}X)$ with dimension $2r$.
To distinguish these dimensions, we use the notation $\dimC(V)=r$ and $\dimR(V)=2r$.\qed
\end{note}

The \emph{dimension} of the cone $C$ is defined as  the dimension of the linear space $\mathbb{R} C=C+(-C)$ spanned by $C$, i.e. $\dim(C):=\dim(\mathbb{R}C)$.
Of course, since a set of generators of a cone $C$ is also a system of generators of the linear space $\mathbb{R} C$, we conclude that the number of rays of a cone is at least its dimension.

\vspace{2mm}

Returning to our setting, we consider the \emph{real} vector space $\LL_{GM}^{\mathbb{R}}$ of dimension $n(n-1)$ spanned by the $n\times n$ elementary rate matrices $L_{ij}, i\neq j$ defined above. We denote 
\[\mathfrak{L}_{GM}^+=\{Q= \sum_{i\neq j}\alpha_{ij}L_{ij} \mid \alpha_{ij}\geq 0\}, \] which  is clearly a convex polyhedral cone in $\LL^{\mathbb{R}}_{GM}$. 
Given a (complex) vector subspace $\LL$ in $\LL_{GM}$, we  consider 
\begin{eqnarray*}
 \LL^+:=\mathfrak{L}\cap \LL^+_{GM}.
\end{eqnarray*}
Notice that all the entries of each matrix in $\LL^+$ are real and non-negative.

\begin{theorem}\label{res:stochastic_cone}
For any (complex) vector subspace $\LL$, 
 $\LL^+=\mathfrak{L}\cap \LL^+_{GM}$ is a strongly convex polyhedral cone in $\LL_{GM}^{\mathbb{R}}$.
 The dimension of $\LL^+$ as a cone is less than or equal to the complex dimension of $\LL$, and equality holds if and only if $\LL$ has a stochastic generating set. 
\end{theorem}

\begin{proof}
The set  $\LL^+$ is the intersection of two convex polyhedral cones, so it is also a convex polyhedral cone. 
Moreover, being contained in $\LL^+_{GM}$ it is clear  it contains no linear subspaces, so it is strongly convex, as required. %
Now, to show that the dimension of $\LL^+$ is less than or equal to the complex dimension of $\LL$, consider the vector space $\mathbb{C}\LL^+$ and observe it is a subspace: $\mathbb{C}\LL^+\subset \LL$. 
This implies $\dimC(\mathbb{C}\LL^+)\leq \dimC(\LL)$, and since $\LL^+$ contains only real vectors, we have $\dimC(\mathbb{C}\LL^+)=\dimR(\mathbb{R}\LL^+)=\dim(\LL^+)\leq \dimC(\LL)$, as required.
Now, assume  $\LL$ has a stochastic generating set $B_{\LL}$ so  $B_{\LL}\subset \LL^+$ and $\mathbb{C}B_{\LL}=\LL$.
As $B_{\LL}$ contains only real vectors, we have $\dimR(\mathbb{R}B_{\LL})=\dimC(\mathbb{C}B_{\LL})=\dimC(\LL)$;
and because $\LL^+$ contains only real vectors and $\LL^+\subset \LL$, we have $B_\LL\subset \LL^+\subset \mathbb{R}B_\LL$, so $\mathbb{R}B_\LL
=\mathbb{R}\LL^+$.
Together this implies $\dimR(\mathbb{R}B_\LL)=\dimR(\mathbb{R}\LL^+)=\dim(\LL^+)=\dimC(\LL)$.
Conversely, suppose  $\dimC(\LL)=\dim(\LL^+)$.
Take a generating set for $\mathbb{R}\LL^+$ composed of vectors in $\LL^+$; by removing vectors in this generating set, we can always assume they actually form a basis $B\subset \LL^+$ of $\mathbb{R}\LL^+$.
Now consider the vector subspace $\mathbb{C}B\subset \LL$ and observe  $\dimC(\mathbb{C}B)=\dimR(\mathbb{R}B)=\dimR(\mathbb{R}\LL^+)=\dim(\LL^+)=\dimC(\LL)$. Thus $\mathbb{C}B=\LL$, as required. \qed
\end{proof}

\begin{remark}
Assume $\LL$ is a Lie algebra without a stochastic basis and take $\LL'=\mathbb{C} \LL^+$ as the complex span of the  cone $\LL^+$. Since $\LL'$ is a complex vector space with a stochastic basis, the result above shows that its complex dimension equals the dimension of $\LL^+$, which is strictly smaller than the dimension of $\LL$. Moreover, notice that $\mathbb{R}\LL^+$  is closed under the Lie bracket, since $\LL$ is a Lie algebra and the Lie bracket of matrices with real entries still has real entries. 
Since $\LL'$ is generated by $\mathbb{R}\LL$, it follows that $\LL'$ is a complex Lie algebra. 
We conclude that if $\LL$ is a Lie algebra with no stochastic basis, then we can construct a strictly smaller Lie algebra $\LL'$ with a stochastic basis such that $(\LL')^+
=\LL^+$. 
This fact justifies Definition~\ref{def:LMM}. 
By requesting the Lie algebra to have a stochastic basis, we are considering the smallest Lie algebra that contains the cone $\LL^+$. 
Otherwise, we would be led to more than one Lie algebra giving rise to the same  cone $\LL^+$. \qed
\end{remark}

 
 \begin{definition}
 The \emph{dimension} of a Lie Markov model is the  dimension of $\mathfrak{L}$ as a complex vector space (which by virtue of Theorem~\ref{res:stochastic_cone} equals the dimension of $\LL^+$ as a cone). The \emph{stochastic cone of $\LL$} is the convex polyhedral cone $\LL^+$ and the \emph{rays of the model} are the rays of $\LL^+$.
 \end{definition}


\begin{remark}
 It is important to note that not every stochastic generating set of $\LL$ is a set of generators of the cone $\LL^+$.  
If this is the case and the set of generators is \emph{minimal}, the positive linear span of each generator is a ray of the cone. \qed
\end{remark}


\subsubsection*{Background on group representation theory}
In what follows we recall basic results from the representation theory of permutation groups $G\leq \SG_n$.
We recommend \cite{sagan2001} and \cite{james2001} as an excellent introductions to the required material.

A (linear) \emph{representation} of a group $G$ is a group homomorphism $\rho:G\rightarrow GL(V)\cong GL(m,\mathbb{C})$, where $V$ is a $\mathbb{C}$-vector space of dimension $m$.
In this situation, $\rho$ provides an \emph{action} of $G$ on $V$, and we say that $V$ forms a $G$-\emph{module}. 
A representation is said to be \emph{irreducible} if it does not contain any proper $G$-submodules.

Let $G\leq  \SG_n$ be a permutation group on $n$ elements.
Write $\{V_i\}_{i=1,\ldots,l}$ for the irreducible $G$-modules and $\rho_i:G\rightarrow GL(V_i)$ for the corresponding group homomorphism.
Since $G$ is finite, any representation $\rho:G\rightarrow GL(V)$ is completely reducible and there is a decomposition of the corresponding module $V$ into irreducible parts called \emph{isotypic components}, so we can write (Maschke's theorem):
\begin{eqnarray}\label{Masch}
V\cong \oplus_{i=1}^\ell  c_i V_i,
\end{eqnarray}
where the $c_i$ are non-negative integers specifying the number of copies of the module $V_i$ in the decomposition of $V$.

\begin{example}
The irreducible representations of $\SG_n$ are indexed by the integer partitions of $n$ \cite{sagan2001}. 
 The \emph{defining} representation of $\SG_n$ is defined on  the vector space $\mathbb{C}^n=\lrc{\{e_i\}_{1\leq i\leq n}}$ by $\sigma: e_i\mapsto e_{\sigma(i)}$.
It decomposes as $\{n\}\oplus \{n-1,1\}$, where $\{n\}$ is the (one-dimensional) trivial representation and $\{n-1,1\}$ has dimension $n-1$.
 \qed
 \end{example}
 
 \begin{example}\label{the_group}
  After identifying the nucleotides $A,G,C,T$ with the integers $1,2,3,4$, consider $\G$ as a subgroup of $\SG_4$:
\begin{eqnarray*}
\G:=\{e,(12),(34),(12)(34),(13)(24),(14)(23),(1324),(1423)\}.
\end{eqnarray*}
 The group $\G$ has 5 conjugacy classes:
%
\begin{eqnarray*}
 \left  [e \right ] & = & \{e\},\\
 \left [(12) \right ]& = &\{(12),(34)\}, \\
 \left [(12)(34)\right ]& = &\{(12)(34)\},  \\  
 \left [(13)(24)\right ]& = &\{(13)(24),(14)(23)\}, \\
 \left [(1324)\right ]& = &\{(1324),(1423)\}.   
    \end{eqnarray*}
Recall  the number of irreducible representations of a finite group is equal to the number of its conjugacy classes, and the sum of the dimension of each irreducible representation
squared is equal to the order of the group (see \cite{sagan2001} for example). 
We conclude  there are five irreducible representations of $\G$, which we denote as $\id$, $\sgn$ $\zeta_1$, $\zeta_2$, and $\xi$; with the corresponding character table presented as Table~\ref{tab: chartabZZ}. 
Notice the first row in the character table gives the dimension of each representation. 
Notice also  there are four one-dimensional representations, namely $\id$ (the trivial representation), $\sgn$ (each permutation $\sigma$ is mapped to $\sgn(\sigma)$), $\zeta_1$ and $\zeta_2$. Besides these, the representation $\xi$ is two-dimensional. 
The rows of Table \ref{tab: chartabZZ} represent the conjugacy classes of~${\G}$. 
\qed
\end{example}
 
Every irreducible module $V_i$ of $G$ has a \emph{projection operator} associated to it:
\begin{eqnarray}\label{proj_op}
 \Theta_{i}(\sigma):=\fra{1}{|G|}\sum_{\sigma \in G}\overline{\chi^{i}(\sigma)}\rho(\sigma),
\end{eqnarray}
where $\chi^i: G\rightarrow \mathbb{C}$ is the \emph{character} of the irreducible representation $\rho_i$ defined as $\chi^i(\sigma):=\text{tr}(\rho_i(\sigma))$, i.e. the trace of the representing matrix $\rho_i(\sigma)$. These operators project a given $G$-module $V$ onto its isotypic components, i.e. $\Theta_{i}(V)=c_i V_{i}$, 
so  they can be used to compute the $c_i$ as well as to identify generators of the components.

Of course, we can restrict $\rho$ to any subgroup  $H\leq G$, 
making a $H$-module of $V$. 
By virtue of Maschke's theorem, we can also decompose $V$ into the irreducible $H$-modules.
Recall an irreducible representation of $G$ does not necessarily stay irreducible when restricted to a subgroup $H$ of $G$. The \emph{branching rule } $G\downarrow H$ applies to describe the decomposition of the irreducible representations of $G$  in terms of the irreducible representations of $H$ (see Chap. 2.8, of \cite{sagan2001}).
By applying orthogonality in the character tables of $\SG_{4}$ and $\G$ (see Table \ref{tab: chartabZZ}), and concentrating on the conjugacy class $[(12)(34)]$ in $\SG_4$ compared to the same class in $\G$, it is straightforward to derive the group branching rules shown in Table \ref{tab: branching_rule}.

\subsubsection*{Background on discrete group actions}
Whenever a group $G$ acts on a finite set $B=\{b_1,\ldots,b_t\}$, there is a group homomorphism 
\begin{eqnarray}\label{act}
\rho:G\rightarrow \SG_t. 
\end{eqnarray}
A $G$-\emph{orbit} in $B$ is a subset $\mathcal{B}=\{b_{i_1},b_{i_2},\ldots , b_{i_{l}}\}\subset B$ which is invariant under $G$ and is  minimal.
That is 
\begin{eqnarray*}
\sigma\mathcal{B}:=\{b_{i_{\rho(\sigma)(1)}},b_{i_{\rho(\sigma)(2)}},\ldots , b_{i_{\rho(\sigma)(l)}}\}=\mathcal{B}, \mbox{ for all }\sigma\in G,
\end{eqnarray*}
and $\mathcal{B}$ contains no smaller subsets with this property.
From this, we can decompose $B$ as a disjoint union of $G$-orbits:
$B=\mathcal{B}_1\cup \mathcal{B}_2 \cup \ldots \cup \mathcal{B}_k$.

%
If we focus on each orbit, the \emph{orbit stabilizer theorem} (see \cite{bogopolski}, for example) states that, up to bijective correspondence, every $G$-orbit has the form of the quotient 
\begin{eqnarray*}
 G/H=\{[\sigma_1],\ldots,[\sigma_q]\}, \qquad [\sigma_i]:=\sigma
_i H,
\end{eqnarray*}
where  $H$ is a subgroup of $G$ and the $\sigma_i\in G$ are chosen so that the coset $\sigma_j H \neq \sigma_i H$ if $i\neq j$.
The group operation of $G$ induces an action in the finite set $G/H$ by 
$\sigma:\sigma_i H\mapsto (\sigma\sigma_i) H$.
Actually,  $H$ is the \emph{stabilizer} of some element $x\in \mathcal{B}$, $G_x:=\{g\in G:gx=x\}$. 
As $G_x\leq G$, and there are only finitely many subgroups of $G$, it is thus possible to give a complete list of $G$-orbits (up to isomorphism) by simply listing all quotients $G/H$ with $H\leq G$.
We recall we can turn the quotient $G/H$ into a $G$-module by considering the vector space generated by the cosets of $G/H$:
\[\lrc{G/H}=\lrc{[e], [\sigma_2], \ldots, [\sigma_q] }=\{v=c_1[e]+c_2[\sigma_2]+\ldots +c_q[\sigma_q]:c_i\in\mathbb{C}\},\] 
with the action
$\sigma:v=\sum c_i[\sigma_i] \mapsto v'=\sum c_i[\sigma\sigma_i]$.
%

Back to the general case, the action of $G$ on the set $\mathcal{B}$ induces a representation of $G$ in the vector space $\mathbb{C}\mathcal{B}$. 
We will refer to these as \emph{permutation representations} and they will play a key role throughout the paper. 
Notice  they decompose as 
\begin{eqnarray*}
 \mathbb{C}\mathcal{B}\cong \oplus_{i=1}^k \langle G/H_i \rangle_{\CC}
\end{eqnarray*}
where $H_i< G$ is the stabilizer of some element in the orbit $\mathcal{B}_i$. 
The reader should note this decomposition is not the decomposition into irreducible representations of (\ref{Masch}). In fact it is possible to show that any permutation representation is actually reducible.

%

\begin{table}[t]
\caption{\label{tab: chartabZZ} \small{The character tables of $\mathfrak{S}_4$ and $\G=\left\{e,(12),(34),(12)(34),(13)(24),(14)(23),(1324),(1423)\right\}$. The rows are labelled by the conjugacy classes and the columns are labelled by the irreducible characters. The character table of a group plays a key role to obtain the decomposition of any representation of that group into its irreducible representations (see (\ref{Masch})).}}
\vspace*{2mm}
\begin{tabular}{ccc}
\begin{tabular}{cccccc}
\hline\noalign{\smallskip}
$\mathfrak{S}_4$ & $\{4\}$  & $\{31\}$ & $\{2^2\}$ & $\{21^2\}$ & $\{1^4\}$ \\
\hline
$e$ & 1 & 3 & 2 & 3 & 1 \\
$[(12)]$ & 1  & 1 & 0 & -1\hspace{.35em} & -1\hspace{.35em}\\
$[(123)]$ & 1 & 0 & -1\hspace{.35em} & 0 & 1  \\
$[(12)(34)]$ & 1 & -1\hspace{.35em} & 2 & -1\hspace{.35em} & 1 \\
$[(1234)]$ & 1  & -1\hspace{.35em} & 0 & 1 & -1\hspace{.35em}\\
\noalign{\smallskip}\hline
\end{tabular}
& \qquad & 
\begin{tabular}{cccccc}
\hline\noalign{\smallskip}
$\G$ & $\id$  & $\sgn$ & $\zeta_1$ & $\zeta_2$ & $\xi$ \\
\hline
$e$ & 1 & 1 & 1 & 1 & 2 \\
$[(12)]$ & 1  & -1 & -1 & 1 & 0 \\
$[(12)(34)]$ & 1 & 1 & 1 & 1 &-2 \\
$[(13)(24)]$ & 1  & 1 & -1 & -1 & 0 \\
$[(1324)]$ & 1  & -1 &  1 & -1 & 0 \\
\noalign{\smallskip}\hline
\end{tabular}
\end{tabular}
\end{table}

\section{Lie Markov models with prescribed symmetry}\label{sec3}

In \cite{LMM}, we learnt  the search for Lie Markov models is significantly simplified by demanding the models to have some non-trivial symmetry since this reduces a potential infinity of models to just a number of special cases. The idea is to rely on imposing symmetry to assist in the search for Lie Markov models.
An alternative strategy would be to enumerate all possible Lie Markov models and toss out those without the desired symmetry. However unless the number of states is equal to 2 or 3, this approach is computationally infeasible.
Thus, we are led to deal with the technicalities of this section in order to refine our search of the Lie Markov models with some prescribed symmetry. 
Of course, it is expected that the larger the symmetry we demand, the easier the analysis will be.  

To this aim, recall  the symmetric group $\mathfrak{S}_n$ has an action on $\mathfrak{L}_{GM}$ defined on the elementary rate matrices as 
$\rho(\sigma) \cdot L_{ij} := L_{\sigma(i)\sigma(j)}$ for all $\sigma\in \mathfrak{S}_n$,
and extended to all of $\mathfrak{L}_{GM}$ by linearity. 
Equivalently, the action can be defined by  \begin{eqnarray}\label{action_GM}
 Q=\sum_{i\neq j}\alpha_{ij}L_{ij}\mapsto \sigma\cdot Q := K_{\sigma} Q K_{\sigma}^{-1}=\sum_{i\neq j}\alpha_{ij}L_{\sigma(i)\sigma(j)},
\end{eqnarray}
where $K_{\sigma}$ is the \emph{permutation matrix} associated to $\sigma$. 

\begin{definition}[cf. \cite{LMM}]\label{def:symm}
We say  a Lie Markov model $\mathfrak{L}$  has the \emph{symmetry} of the group $G\leq\SG_n$ if there is a basis $B_{\mathfrak{L}}$ of $\LL$ invariant under the  action of $G$ induced by (\ref{action_GM}), that is, a basis $B_{\mathfrak{L}}=\{L_1,L_2,\ldots, L_d\}$ such that 
\begin{eqnarray}
\sigma\cdot B_{\mathfrak{L}}:=\left\{K_{\sigma}L_{1}K^{-1}_{\sigma},K_{\sigma}L_{2}K^{-1}_{\sigma},\ldots ,K_{\sigma}L_{d}K^{-1}_{\sigma}\right\}=B_{\mathfrak{L}},\quad\forall \sigma\in G.\nonumber
\end{eqnarray} 
In this case, we will say that $B_{\LL}$ is a \emph{permutation basis} of $\LL$.
\end{definition}

Notice  a Lie Markov model $\LL$ has the symmetry of $G$ if and only if there is a permutation representation of $G$ on $\LL$, so  we have a decomposition $\LL\cong \oplus_{i=1}^k \langle G/H_i\rangle_{\CC}$. A permutation basis for $\LL$ is then obtained by collecting a permutation basis $B_i$ for each $\langle G/H_i\rangle_{\CC}$ and putting them together.

%
\begin{remark}
 Notice  if $\mathfrak{L}$ has the symmetry of a permutation group $G$, then it also has 
 the symmetry of any subgroup $H\leq G$. \qed
\end{remark}
%
%

The reader is referred to \cite{LMM} for the statistical motivations for this definition.
In a nutshell, parameter estimation under such a model is invariant under nucleotide permutations belonging to $G$.
In particular, we have a group homomorphism 
\begin{eqnarray}\label{eq:homomorphism_basis}
 \rho:G \rightarrow \mathfrak{S}_d,
\end{eqnarray}
where the image of any permutation $\sigma\in G$ is determined by the equality $K_{\sigma}L_{i}K^{-1}_{\sigma}=L_{\rho(\sigma)(i)}$.
Thus, for any rate matrix $Q=\sum_{i=1}^d\alpha_iL_i\in \mathfrak{L}$, we have 
\begin{eqnarray}
\sigma: Q=\sum_{i=1}^d\alpha_iL_i\mapsto \sum_{i=1}^d\alpha_iL_{\rho(\sigma)(i)}=\sum_{i=1}^d\alpha_{\rho(\sigma^{-1}(i))}L_{i},\nonumber
\end{eqnarray}
so $G$ acts by permuting the model parameters, ie. $\alpha_i\mapsto \alpha_{\rho(\sigma^{-1}(i))}$, and hence leaves maximum likelihood estimates invariant.
%

\begin{example}\label{S4-LMM}
\cite{LMM}
The list of 4-state Lie Markov models with $\SG_4$ symmetry is: 
\begin{enumerate}
\item Jukes-Cantor model, with dimension 1 \cite{JC69};
\item Kimura model, with dimension 3 \cite{kimura1981};
\item Felsenstein model, with dimension 4 \cite{felsenstein1981};
\item Kimura+Felsenstein model or ``K3ST+F81'', with dimension 6 (see \cite{sumner2012}, and Example~\ref{Ex:F81+K3ST} below);
\item General Markov model, with dimension 12.
\end{enumerate}
\end{example}

Presently, we recall the general procedure to obtain Lie Markov models with prescribed symmetry.
Suppose we have a Lie Markov algebra $\mathfrak{L}$ with dimension $d$ and a permutation group $G\leq \SG_n$. 
We demand that $\mathfrak{L}$ satisfies the conditions of Definition~\ref{def:symm} for the permutation group $G$. 
Then, $\LL$ is provided with a basis $B_{\LL}$ which is invariant under $G$. 
As explained above, we have a decomposition of $B_{\LL}$ into $G$-orbits. 
We can then compare the irreducible $G$-modules that occur in the decomposition of $\mathfrak{L}_{GM}$ to those that occur in the decomposition of $\lrc{G/H}$ for each $H\leq G$.
Finally, we can attempt to construct subalgebras $\mathfrak{L}\subset \mathfrak{L}_{GM}$ with a basis $B_\mathfrak{L}$ such that $B_\mathfrak{L}
=\mathcal{B}_1\cup \mathcal{B}_2 \cup \ldots \cup \mathcal{B}_r$ is a plausible union of orbits $\mathcal{B}_i$ consistent with the linear decomposition of $\mathfrak{L}_{GM}$ induced by the action of $G$.

\begin{table}[t]
\caption{\label{tab: branching_rule} \small{The branching rule of $\SG_4\downarrow \G$  describes the decomposition of the irreducible representations of $\SG_4$ when restricted to the subgroup $\G$. For example, $\{2^2\} \mapsto \id+\sgn$ means that, when restricted to $\G$, the irreducible representation $\{2^2\}$ of $\SG_4$ decomposes as one copy of the \emph{identity} representation of $\G$, and one copy of the \emph{sign} representation of $\G$.}}
\vspace*{2mm}
\centering
\begin{tabular}{ll}
\hline\noalign{\smallskip}
 & $\{4\}  \mapsto \id$ \\
& $\{1^4\} \mapsto \sgn$ \\
$\SG_4 \downarrow {\G}$ :  & $\{31\}  \mapsto  \xi+\zeta_2$ \\
& $\{2^2\}  \mapsto  \id+\sgn$ \\
& $\{21^2\}  \mapsto  \xi+\zeta_1$ \\ 
\hline\noalign{\smallskip}
\end{tabular}
\end{table}

The general procedure is: 
\begin{enumerate}
\item Decompose the Lie algebra of the GM model into irreducible modules of $G$:
\begin{eqnarray}\label{dec:L_GM}
\mathfrak{L}_{GM}=\oplus_{k}f_kV_{k}, 
\end{eqnarray}
where $k$ labels the irreducible $G$-module $V_{k}$ and the $f_k$ are non-negative integers specifying the number of copies  of each irreducible module in the decomposition.

\item Apply the orbit stabilizer theorem and construct the list of $G$-orbits, $G/H$, by working through the subgroups $H\leq G$. 
For each subgroup $H$, extend the orbits linearly over $\mathbb{C}$ to the $G$-module $\lrc{G/H}$ and decompose this space into irreducible $G$-modules: 
\begin{eqnarray*}
 \lrc{G/H}\cong \oplus_{k}b^{H}_{k}V_{k},
\end{eqnarray*}
where again the $b^{H}_{k}$ are non-negative integers.

\item Working up in dimension $d$, consider all unions of $G$-orbits
$S=\bigcup_{i=1}^q(G/H_i)$
such that $\left|S\right|=\sum_{1\leq i\leq q}\left|G/H_i\right|= d$ (where $|\cdot|$ stands for cardinality).
For each $S$, consider its linear decomposition into irreducible $G$-modules  
\begin{eqnarray*}
\lrc{S}\cong\oplus_k a_k V_{k} 
\end{eqnarray*}
where $a_k:=b^{H_1}_{k}+b^{H_2}_{k}+\ldots +b^{H_q}_{k}$, and, in order to exclude unions of $G$-orbits that do not occur in the linear decomposition of $\mathfrak{L}_{GM}$ as a $G$-module, check  $a_k\leq f_k$, for each $k$. 
 
\item For each case thus identified, consider the vector space $\mathfrak{L}=\oplus_{k}a_{k} V_{k}$ and use explicit computation to check whether $\mathfrak{L}$ forms a Lie algebra.
If so, attempt to show  it has a stochastic basis. 
\end{enumerate}

This procedure is guaranteed to produce all Lie Markov models with symmetry $G$. 
In Section~\ref{sec5}, we will give a complete presentation of the 4-state models with purine/pyrimidine symmetry. 

\begin{remark}
 In our procedure we first look for all possible decompositions into irreducible modules for a permutation representations and  we investigate how these decompositions are realised into Lie subalgebras of $\LL_{GM}$. A different approach would be to deal first with possible Lie subalgebras of $\LL_{GM}$ (up to isomorphism) and then, for each isomorphism class, look for possible subalgebras which are permutation representations of $G$. Our experience tells us  this second part is rather unfeasible and, in the last section, we adopt the procedure just explained. \qed
\end{remark}

\begin{remark}
\emph{Equivariant} models were first introduced in  \cite{draisma2008} and have been studied in \cite{casfer2010}. 
In \cite{LMM} we modified slightly the definition of equivariant models to adapt it to the continuous-time Markov model setting. 
Under this definition, equivariant models appear as a particular case of Lie Markov models. 
Actually, the $G$-equivariant model is the Lie Markov model with $G$ symmetry obtained when we take $\LL$ to be the isotypic component of $\LL_{GM}$  associated to the \emph{trivial} or \emph{identity} representation $\id$ (which maps each permutation to the identity map): $\LL=f_{\id} V_{\id}$   (see (\ref{dec:L_GM})).
For example, the $\SG_4$-equivariant model is the Lie Markov model with symmetry $\SG_{4}$ and decomposition $\LL\cong \id$: it is the Jukes-Cantor model \cite{JC69}. 
In a similar way, we will recover the Kimura model with two parameters \cite{kimura1980} as the Lie Markov model with symmetry $\G$ and  decomposition $\LL \cong 2 \id$ (see Model 2.2b in Section~\ref{sec5}). \qed
\end{remark}

\subsection*{\emph{The stochastic cone of a Lie Markov model}}

We want to explore the geometry of the stochastic cone associated to a Lie Markov model with symmetry given by some permutation group $G\leq \SG_n$. 
Since the action of $G$ on $ \mathfrak{L}_{GM}$ is as given in (\ref{action_GM}), we infer that the space
 $ \LL_{GM}^+$ is invariant under this action, i.e. 
$ G \mathfrak{L}_{GM}^+= \mathfrak{L}_{GM}^+$.
 From this, we conclude that if $\mathfrak{L}\subset \mathfrak{L}_{GM}$ is a vector subspace which is invariant under the action of $G$, then the stochastic cone 
$ \mathfrak{L}^+=\mathfrak{L}\cap \mathfrak{L}^+_{GM}$
is invariant under $G$ as well. 

Because each permutation in $G$ induces a linear automorphism in $\mathfrak{L}_{GM}$ and the cone $\mathfrak{L}^+$ is invariant, the set of rays of the cone must also be invariant under the action of $G$. 
We infer that, after giving an ordering to the set of  rays, there is a group homomorphism
\begin{eqnarray}\label{homomorphism}
 G\rightarrow \mathfrak{S}_r,
\end{eqnarray}
where $r$ is the number of rays of $\mathfrak{L}^+$. 
From this, we can decompose the set of rays of $\LL^+$ into $G$-orbits, which we will refer to as \emph{ray-orbits}. 
Notice in general, the above homomorphism is different from the homomorphism arising from a permutation basis, as described in (\ref{eq:homomorphism_basis}).

\begin{example}
The number of rays of $\LL_{GM}^+$ is $n(n-1)$. 
These rays are exactly the positive span of the elementary rate matrices $L
_{ij}$. 
The group homomorphism  $G\rightarrow \mathfrak{S}_r$ of (\ref{homomorphism}) corresponds to the action described in (\ref{action_GM}).\qed
\end{example}

%
%

\begin{example}\label{Ex:F81+K3ST}
In Result 17 of \cite{LMM}, we learnt there is only one six-dimensional Lie Markov model with $\SG_4$ symmetry. 
The Lie algebra  $\LL$ is the vector space sum of the Kimura 3ST and Felsenstein 1981 models. It is generated by 
\begin{eqnarray*}
 W_{ij}=L_{s(ij)}+(R_i+R_j), \qquad i<j, \quad i,j \in \{1,2,3,4\}, 
\end{eqnarray*}
where $R_i=\sum_{j\neq i}L_{ij}$ and $L_{s(ij)}=L_{ij}+L_{ji}+L_{kl}+L_{lk}$ with $i,j,k,l$ all different. 
The reader may notice, although the 6 vectors $W_{ij}$ do form a permutation basis of $\LL$, by taking the convex cone generated by them, 
$ \{\sum \lambda_{ij} W_{ij} \mid \lambda_{ij}\geq 0 \}$,
we are not considering all the stochastic rate matrices in the model.
For example, the vector $R_1$ is in the stochastic cone $\LL^+$ but we cannot obtain it as a positive linear combination of the vectors $W_{ij}$.

The reader may argue this situation occurs because of our particular choice of a permutation basis, but  this will be the case no matter the permutation basis of $\LL$ we consider. 
Actually, the stochastic cone $\LL^+$ has seven rays $\{L_{\alpha},L_{\beta},L_{\gamma},R_1,R_2,R_3,R_4 \}$ (with the notation used there: $L_{\alpha}=L_{s(12)},L_{\beta}=L_{s(13)},L_{\gamma}=L_{s(12)}$). 
 We will find this model again in Section~\ref{sec5} of this paper as Model 6.7a. \qed
 \end{example}

\section{Decomposition of $\LL_{GM}$ as a $\G$-module }\label{sec4}

As we are especially interested in nucleotide evolution, we fix $n\!=\!4$ and deal with the group of permutations that preserves the partitioning of nucleotides into purines and pyrimidines: $AG|CT:=\{\{A,G\},\{C,T\}\}$.

By identifying nucleotides $\{A,G,C,T\}$ with numbers $\{1,2,3,4\}$, this leads to consider the subgroup  $\G$ of $\SG_4$ presented in Example \ref{the_group}:
\begin{eqnarray*}
\G:=\{e,(12),(34),(12)(34),(13)(24),(14)(23),(1324),(1423)\}.
\end{eqnarray*}
Of course, we expect to recover the Kimura model with 3 parameters in the list of Lie Markov models with this symmetry: it is Model 3.3a. However, as already noted in Remark \ref{S4-LMM}, this model has a wider symmetry and in fact, it is  $\mathfrak{S}_4$-symmetric (see \cite{LMM}). 

Presently, we use the projection operators to decompose the Lie algebra of the general Markov model into the irreducible representations of $\G$.

\begin{remark}
 The reader may notice the irreducible characters of the group $\G$ in Table~\ref{tab: chartabZZ} take only real values. As a consequence, irreducible real representations remain irreducible over the complex field and all the representation theory for $\G$ can be dealt over the real field. However, we prefer to keep our study over the complex as this is the field where the general theory is developed. For instance, it is important to work over the complex field when computing the full list of $\G$-submodules of $\LL_{GM}$ isomorphic to $\oplus_k a_k V_k$ (see the step 4 of the procedure of Section 3): to this aim, we apply that the only $\G$-endomorphisms of an irreducible module are of the form $\lambda \1$ (Schur's lemma), which is known to be false if the field is not algebraically closed. \qed
\end{remark}

From now on, we will consider the restriction of the action of $\SG_4$ described in (\ref{action_GM}) to the group ${\G}$. 
We will denote this action by $\rho_{\G}$:
\begin{eqnarray}\label{action_G}
 \rho_{\G}(\sigma): Q\mapsto K_{\sigma} Q K_{\sigma}^{-1}.
\end{eqnarray}
%
In Result 8 of \cite{LMM}, we learnt  the decomposition of  the $\LL_{GM}$ into the irreducible representations of $\SG_4$ (expressed using integer partitions of 4) is 
$\LL_{GM}\cong \{4\}\oplus 2\{31\} \oplus \{2^2\}  \oplus \{21^2\}$.
%
By applying the branching rule of $\SG_4$ to $\G$ (see Table \ref{tab: branching_rule}) we obtain:

\begin{theorem}\label{dec_GM}
The decomposition of the 4-state general rate matrix model $\mathfrak{L}_{GM}$ into irreducible representations of ${\G}$ is given by 
\begin{eqnarray}\label{eq:schurdecompG}
\mathfrak{L}_{GM}\cong 2\, \id\oplus  \sgn \oplus \zeta_1 \oplus 2\, \zeta_2 \oplus 3\, \xi,
\end{eqnarray}
where the decomposition of the dimension is given by 
$12=2\times 1+1+1+2\times 1 +3\times 2$.
\end{theorem}

\subsection*{\emph{Decomposition of the orbits of $\G$ in $\LL_{GM}$}}

Following the general scheme described in Section~\ref{sec3}, our task now is to identify the Lie Markov models occurring as subalgebras of $\LL_{GM}$ and with symmetry ${\G}$.
In Table~\ref{tab:Z2models} we present the decomposition of the orbits of ${\G}$.
These are computed by using the orbit stabilizer theorem and projecting $\lrc{{\G}/H}$ onto the irreducible modules $V_i$ of ${\G}$ using the projection operators $\Theta_i$ defined in (\ref{proj_op}).

\begin{example}\rm 
Here we develop the case of $\HH=\{e,(12)(34)\}$ as an illustrative example. 
We have ${\G} \, / \, \HH=\left\{ [e],[(12)],[(13)(24)],[(1324)] \right \}$,
where $[\sigma]$ represents the coset in ${\G}/\HH$ containing the element $\sigma$. Namely, 
${[e]}= \{e,(12)(34)\}$, ${[(12)]}= \{(12),(34)\}$, 
${[(13)(24)]} = \{(13)(24),(14)(23)\}$ and ${[(1324)]} = \{(1324),(1423)\}$.
These cosets inherit an action of ${\G}$ by taking $\sigma: [\sigma']\mapsto [\sigma\sigma']$, which can be extended linearly to a linear representation of ${\G}$ by taking the module 
$\lrc{{\G}\, / \,\HH}\cong \CC^4.$
Next, we decompose $\lrc{{\G}\, / \,\HH}$ into irreducible modules of ${\G}$ by applying the projection operators: $\Theta_{\id}$, $\Theta_{\sgn}$, $\Theta_{\zeta_1}$, $\Theta_{\zeta_2}$ and $\Theta_{\xi}$.
For example:
\begin{eqnarray*}
\Theta_{\id}[e]= \fra{1}{8}\sum_{\sigma\in {\G}}\sigma\cdot [e]
=\fra{1}{4}\left([e]+[(12)]+[(13)(24)]+[(1324)]\right).
\end{eqnarray*}
As this projection is non-zero, we conclude  $\lrc{{\G}/\HH}$ contains the trivial representation $\id$. 
We can check that the image by $\Theta_{\id}$ of the other coset elements gives the same projection, 
so $\lrc{{\G}/\HH}$ contains $\id$ only \emph{once}. 
Similarly, referring to the character table of $\SG_{4}$ (see Table~\ref{tab: chartabZZ}), we have
\begin{eqnarray*}
\Theta_{\sgn}[e]&= & \fra{1}{8}\sum_{\sigma\in {\G}}\chi^{\sgn}(\sigma) \sigma\cdot [e] = \fra{1}{4}\left([e]-[(12)]+[(13)(24)]-[(1324)]\right ),
\end{eqnarray*}
and we check that $\Theta_{\sgn}[e]=\Theta_{\sgn}[(12)]=\Theta_{\sgn}[(13)(24)]=\Theta_{\sgn}[(1324)]$ to learn that $\lrc{{\G}/\HH}$ does contain a copy of the  $\sgn$ representation.
Similarly, we check  $\langle {\G}/\HH \rangle_{\mathbb{C}}$ contains a copy of $\zeta_1$ and $\zeta_2$ representations. 
On the other hand, we see 
\begin{eqnarray*}
\Theta_{\xi}[e]&= & \fra{1}{8}\sum_{\sigma\in {\G}}\chi^{\xi}(\sigma)\cdot [e] =\fra{1}{4}\left([e]-[(12)(34)]\right)=0,
\end{eqnarray*} 
and we check  $\Theta_{\xi}[(12)]=\Theta_{\xi}[(13)(24)]=\Theta
_{\xi}[(1324)]=0$ to learn  $\lrc{{\G}/\HH}$ does not contain a copy of the  $\xi$ representation.
%
Putting this together and counting dimensions, we infer that
$\lrc{{\G}/\HH}\cong \id \oplus \sgn \oplus \zeta_1 \oplus \zeta_2.
$\qed
\end{example}

Proceeding as in this example, we have produced the results summarised in Table~\ref{tab:Z2models}.
It gives the decomposition of $\lrc{{\G}/H}$ into irreducible representations for each subgroup $H\leq {\G}$.
The first column shows how many copies of each $H$ occur  as a subgroup in ${\G}$, with automorphism classes accounted for with distinct decomposition in the fourth column.
For example, there are three automorphism classes classes of  $\mathbb{Z}_2$ in ${\G}$:
$\{e,(12)\}   \cong  \{e,(34)\}$, $ \{e,(12)(34)\}$ and $\{e,(13)(24)\}  \cong  \{e,(14)(23)\}$, and the corresponding spaces $\lrc{{\G}/H}$ have different decomposition into irreducible modules, as shown in Table~\ref{tab:Z2models}.
Similarly, there are two ``types'' of $\SG_2\times \SG_2$:
$\{e,(12),(34),(12)(34)\}$ and $\{e,(12)(34),(13)(24),(14)(23)\}$.
Again, these two types have differing decomposition into irreducible subspaces. 

Finally, Table \ref{tab:all_decs} shows all possible decompositions for a $\G$-invariant subspace of $\LL_{GM}$ allowed by the decomposition of $\LL_{GM}$ of Theorem \ref{dec_GM}.  The list is obtained by adding decompositions of $\G$-orbits (see Table \ref{tab:Z2models}) as long as they are allowed by the decomposition of $\LL_{GM}$ as a $\G$-module (see Theorem \ref{dec_GM}). 
Note  the decomposition (\ref{eq:schurdecompG}) of $\mathfrak{L}_{GM}$ has two copies of the trivial representation while the decomposition of each ${\G}/H$ has only one copy.

Referring to Table~\ref{tab:all_decs}, we conclude:

\begin{theorem}
There are no Lie Markov models with purine/pyrimidine symmetry of dimension seven or eleven.
\end{theorem}

\begin{remark}\label{ray_orbits}
Being a $G$-orbit, we can consider the abstract vector space generated by any ray-orbit $B=\{Q_1,\ldots,Q_r\}$, that is, 
$\{\sum_{i=1}^r a_i [ Q_i ] : a_i\in \CC  \}$,
where the notation $[Q_i]$ is used to emphasise the fact that we are avoiding any reference to matrix addition between the elements of the ray-orbit. The dimension of this vector space equals the number of elements in the orbit, and as a permutation representation, the decomposition into irreducible representations will be one of the decompositions shown in Table \ref{tab:Z2models}. On the other hand, we can also consider the vector subspace  of $\LL_{GM}$ spanned by the matrices $Q_1,\ldots,Q_r$. Notice that these matrices may not be linearly indepenedent as vectors of $\LL_{GM}$ and the dimension of this vector subspace will be smaller than the number of them. In this case, this vector space is not a permutation representation and its decomposition  into irreducible representations does not appear in Table  \ref{tab:Z2models}. 
For an example of this, the reader is referred to ray-orbits $\trip{4}{\frac{1}{3}}{\frac{2}{3}}d, \trip{4}{\frac{1}{3}}{\frac{2}{3}}e,\trip{4}{\frac{1}{3}}{\frac{2}{3}}f$ presented in Table~\ref{tab:ray_orbits}. \qed
\end{remark}

\begin{table}[t]
\newcolumntype{C}{ >{\centering\arraybackslash} c}
\small
\caption{\label{tab:Z2models} \small{Decomposition of the orbits of $\G$ into irreducible modules. }}
\vspace*{2mm}
\centering
{\renewcommand{\arraystretch}{2}
\begin{tabular}{|C|c|c|}
\hline
{automorphism classes  of $H\leq {\G}$} & $ \frac{|{\G}|}{|H|}$ & Decomp. of $\lrc{{\G} / H}$ \\
\hline \hline
$\{e\}$  & 8 & $(1): \id \oplus \sgn \oplus \zeta_1 \oplus \zeta_2  \oplus 2 \xi$  \\  \hline
$\SG_2\cong \{e,(12)\}  \cong \{e, (34)\} $ & 4 & $(2): \id  \oplus \zeta_2 \oplus \xi$  \\  \hline
$\SG_2   \cong   \{e,(14)(23)\}\cong  \{e,(13)(24)\} $  & 4 & $(3): \id  \oplus \sgn \oplus \xi$  \\  \hline
$\SG_2 \cong \{e,(12)(34)\}$  & 4 & $(4): \id  \oplus \sgn \oplus \zeta_1 \oplus \zeta_2$  \\  \hline
$\mathbb{Z}_4\cong \{e,(1324),(12)(34),(1423)\}$  & 2 & $(5): \id \oplus \zeta_1$
\\   \hline
$\SG_2\times \SG_2  \cong \{e, (12),(34),(12)(34)\} $  & 2 & $(6): \id \oplus \zeta_2 $  \\  \hline
$\SG_2\times \SG_2 \cong \{e,(12)(34),(13)(24),(14)(23)\} $  & 2 & $(7): \id \oplus \sgn $  
\\   \hline
${\G}$  & 1 & $(8): \id$ 
\\ \hline
\end{tabular}}
\end{table}

\subsection*{\emph{A convenient basis}}

In this section we derive a basis for the vector space of $4\times 4$ rate matrices $\LL_{GM}$ where the matrices comprising the basis are organised naturally into subsets that span each of the irreducible components of the decomposition of $\LL_{GM}$ with respect to $\G$ (as given in Theorem~\ref{dec_GM}).
This basis is presented in Theorem~\ref{specialbasis} below.
The reader should note the basis vectors we derive play the role the $L_{ij}$ did when we considered models with $\mathfrak{S}_4$ symmetry in \cite{LMM}. 

\begin{table}[t]
\caption{\label{tab:all_decs} \small{Decompositions into irreducible modules of all possible $\G$-permutation subrepresentations of $\LL_{GM}\cong 2 \id \oplus \sgn \oplus \zeta_1\oplus 2 \zeta_2 \oplus 3 \xi$ (see Theorem \ref{dec_GM}). }}
\vspace*{2mm}
\footnotesize
\centering
{\renewcommand{\arraystretch}{1.2}
\begin{tabular}{|c|p{1cm}|c|}
\hline 
Dim.  & Orbits &  Decomp. into irreps.\\
\hline \hline
1 & (8) & $\id$ \\ \hline
2 & 2(8)& $ 2\id $\\
 &  (7) &$ \id \oplus \sgn $\\
 &  (5) &$ \id \oplus \zeta_1 $\\
 &  (6) &$ \id \oplus \zeta_2 $\\\hline
3 &  (7)+(8) &$ 2\id \oplus \sgn $\\
 & (5)+(8) &$ 2\id \oplus \zeta_1 $\\
 & (6)+(8) &$ 2\id \oplus \zeta_2 $\\\hline
4 &  (2)&$ \id \oplus \zeta_2 \oplus \xi $\\
&  (5)+(7) &$2 \id \oplus \sgn \oplus \zeta_1$ \\
 &  (4) &$ \id \oplus \sgn \oplus \zeta_1 \oplus \zeta_2 $\\
 &  (3) &$ \id \oplus \sgn \oplus \xi $\\
 &  2(6) &$ 2\id \oplus 2\zeta_2 $\\
 & (6)+(7) &$ 2\id \oplus \sgn \oplus \zeta_2 $\\
 & (5)+(6) &$ 2\id \oplus \zeta_1 \oplus \zeta_2 $\\\hline
\end{tabular}
\qquad 
\begin{tabular}{|c|p{1cm}|c|}
\hline
Dim.& Orbits &  Decomp.into irreps.\\
\hline \hline
5 & (3)+(8)  &$ 2\id \oplus \sgn \oplus \xi $\\
 &  (4)+(8)&$ 2\id \oplus \sgn \oplus \zeta_1 \oplus \zeta_2 $\\
 & (2)+(8) &$ 2\id \oplus \zeta_2 \oplus \xi $\\\hline
6 & (5)+(3) &$ 2\id \oplus \sgn \oplus \zeta_1 \oplus \xi $\\
 & (4)+(6) &$ 2\id \oplus \sgn \oplus \zeta_1 \oplus 2\zeta_2 $\\
 & (2)+(7) &$ 2\id \oplus \sgn \oplus \zeta_2 \oplus \xi $\\
 & (2)+(5) &$ 2\id \oplus \zeta_1 \oplus \zeta_2 \oplus \xi $\\
 & (2)+(6) &$ 2\id \oplus 2\zeta_2 \oplus \xi $\\\hline
8 & (1) &$\id  \oplus \sgn  \oplus \zeta_1 \oplus \zeta_2 \oplus 2\xi$ \\
& 2(2) &$ 2\id \oplus 2\zeta_2 \oplus 2 \xi $\\
 & (2)+(4) &$ 2\id \oplus \sgn \oplus \zeta_1 \oplus 2\zeta_2 \oplus \xi $\\
 & (2)+(3) &$ 2\id \oplus \sgn \oplus \zeta_2 \oplus 2 \xi $\\\hline
9 & (1)+(8) &$ 2\id \oplus \sgn \oplus \zeta_1 \oplus \zeta_2 \oplus 2 \xi $\\\hline
10 & (1)+(6)  &$ 2\id \oplus \sgn \oplus \zeta_1 \oplus 2\zeta_2 \oplus 2 \xi $\\\hline
12 & (1)+(2) &$ 2\id \oplus \sgn \oplus \zeta_1 \oplus 2\zeta_2 \oplus 3 \xi $\\ \hline
\end {tabular}
}
\end{table}
\normalsize

\paragraph*{Permutation vectors}  
 \quad \newline 
For each $\sigma\in \G, \sigma\neq e$, a \emph{permutation vector} is defined as 
\begin{eqnarray*}
\P_{\sigma}=-\1+K_{\sigma}=\sum_{1\leq j\leq 4} L_{j\sigma(j)}.
\end{eqnarray*}
%
 Notice that each $\P_{\sigma}$ is a rate matrix in $\LL_{GM}$. 
The linear span of these vectors has dimension 5 because of the linear dependencies 
$\P_{(12)}+\P_{(34)}= \P_{(12)(34)}$, and $ \P_{(1324)}+\P_{(1423)}=\P_{(13)(24)}+\P_{(14)(23)}$.
%
 %
Moreover, the permutation vectors span a Lie algebra (Proposition 4.12 of \cite{LMM}):
\begin{eqnarray}
\lie{\P_\sigma}{\P_{\sigma'}}=\lie{-\1+K_{\sigma}}{-\1+K_{\sigma'}}=\lie{K_{\sigma}}{K_{\sigma'}}=K_{\sigma\sigma'}-K_{\sigma'\sigma}=\P_{\sigma\sigma'}-\P_{\sigma'\sigma}.\nonumber
\end{eqnarray}
The permutation vectors are useful because they provide simple expressions of generators of $\LL_{GM}$ consistent with the decomposition of Theorem \ref{dec_GM}. %
The action $\rho_{\G}$ of ${\G}$ on these permutation vectors is given by $\tau: \P_{\sigma} \mapsto K_{\tau} \P_{\sigma} K_{\tau}^{-1}= \P_{\tau \sigma \tau^{-1}}$.  
Notice this action maps each matrix $\P_{\sigma}$ to $\P_{\sigma'}$, where $\sigma'$ is some permutation in the conjugacy class of $\sigma$.
It follows that the vectors $\{L_{\sigma}: \sigma'\in [\sigma]\}$ span a $\G$-module, and by applying character theory we can obtain the decomposition of these $\G$-modules into isotypic components. Moreover, a basis for these $G$-modules consistent with these decompositions can be described with the assistance of the projection 
operators. The following example illustrates this procedure.

\begin{example}
Consider the 2-dimensional subspace $S=\langle L_{(12)}, L_{(34)}\rangle_{\mathbb{C}}$ corresponding to the conjugacy class $[(12)]=\{(12),(34)\}$, and the representation $\rho_{\G}:\G\rightarrow GL(S)$ induced by the action just defined. It is straightforward to check $\tau (12) \tau^{-1}=(12)$, $\tau (34) \tau^{-1}=(34)$ if $\tau\in \{e,(12),(12)(34)\}$, while 
  $\tau (12) \tau^{-1}=(34)$, $\tau (34) \tau^{-1}=(12)$ if $\tau\in \{(13)(24),(1324)\}$. Adopting matrix notation, we obtain 
  \begin{eqnarray*}
   \rho_\G(e)=   \rho_\G\left ( (12) \right )=    \rho_\G\left ((12)(34)\right )=\left ( \begin{array}{cc} 1 & 0 \\ 0 & 1 \end{array} \right) \qquad \mbox{ and } \qquad 
   \rho_\G\left ((13)(24) \right )=
   \rho_\G\left ((1324) \right )=\left ( \begin{array}{cc} 1 & 0 \\ 0 & 1 \end{array} \right). 
  \end{eqnarray*}
 If $\chi$ denotes the character associated with $\rho_G$, we infer 
 \begin{eqnarray*}
 \chi(e)=\chi \left ( (12) \right )=    \chi \left ((12)(34)\right )=2, \qquad \mbox{ and } \qquad 
   \chi\left ((13)(24) \right )=
   \chi\left ((1324) \right )=0.
 \end{eqnarray*}
 By virtue of the character table of $\G$ (see Table \ref{tab: chartabZZ}), we infer $S\cong \id \oplus \xi_2$, and applying the projection operators (see (\ref{proj_op})):
 \begin{eqnarray*}
  \Theta_{\id}(L_{(12)})=\Theta_{\id}(L_{(34)})=\frac{1}{2}(L_{(12)}+L_{(34)});\\
  \Theta_{\xi_2}(L_{(12)})=\Theta_{\xi_2}(L_{(34)})=\frac{1}{2}(L_{(12)}-L_{(34)}).
 \end{eqnarray*}\qed 
 \end{example}

%
Proceeding in this way for each conjugacy class of $\G$ (excluding the trivial class), we identify the following $\G$-modules and decompositions: 
\begin{center}
\begin{tabular}{ll}
  $\lrc{ \P_{(12)},\P_{(34)} }\cong \id \oplus \zeta_2$, & \qquad
 $\lrc{  \P_{(12)(34)} } \cong  \id$, \\
 $\lrc{  \P_{(13)(24)},\P_{(14)(23)} } \cong \id \oplus \sgn$,  & \qquad
 $\lrc{  \P_{(1324)},\P_{(1423)} } \cong  \id \oplus \zeta_1$.
\end{tabular}
\end{center}
%
%
For future convenience, we keep the vectors obtained by applying the projection operators to these decompositions. From now on, we will use the following notation

\begin{center}
\begin{tabular}{lll}
$\B^{\id}_1 = \P_{(12)(34)}$, & $\B^{\id}_2  =  \P_{(13)(24)}+\P_{(14)(23)}$, & $\B^{\sgn}  =  \P_{(13)(24)}-\P_{(14)(23)}$,\\
$\B^{\zeta_1}  =  \P_{(1423)}- \P_{(1324)}$ & $\B^{\zeta_2}_1
=  \P_{(12)}- \P_{(34)}$; 
\end{tabular}
\end{center}
where the superscript indicates which irreducible $\G$-module each vector belongs to.


%


\paragraph*{Cherry vectors} 
\quad \newline 
Referring to Table~\ref{tab:Z2models} and the permutation representation spanned by the ``cherries'' $\{1,2\}$ and $\{3,4\}$, we introduce the matrices
\begin{eqnarray*}
\Q{12}=L_{13}+L_{14}+L_{23}+L_{24},\\
\Q{34}=L_{31}+L_{32}+L_{41}+L_{42},
\end{eqnarray*}
and obtain $ \lrc{ \Q{12},\Q{34} } \cong \id \oplus \zeta_2$.
The action of ${\G}$ on each of these vectors is given by  
$\tau: \Q{ij} \mapsto \Q{\tau(i)\tau(j)}$.
Notice that $\Q{12}+\Q{34}=\B^{\id}_2$. %
By applying the projection operator $\Theta_{\zeta_2}$, we see that $\B^{\zeta_2}_2:=\Q{12}-\Q{34}$ accounts for the second copy of $\zeta_2$. 
%



\paragraph*{Row-sum and twisted vectors}
\quad \newline 
Keeping the notation of \cite{LMM}, define the \emph{row-sum} vectors
\begin{eqnarray}
R_{i}:=\sum_{{j: 1\leq i\neq j\leq 4}}L_{ij}.\nonumber
\end{eqnarray}
The action $\rho_{\G}$ of ${\G}$ on each of these is
$\sigma:R_{i} \mapsto R_{\sigma(i)}$,
and it is isomorphic to the restriction of the defining representation of $\SG_4$ to ${\G}$.
Therefore, the (invariant) subspace generated by the row-sum vectors  is isomorphic to $\id\oplus\{31\}$, restricted to the subgroup ${\G}$. 
By applying the branching rule $\SG_4\downarrow \G$ given in Table \ref{tab: branching_rule}, we obtain 
\begin{eqnarray*}\label{eq:Rdecomp}
\lrc{ R_1,R_2,R_3,R_4}\cong \id \oplus \zeta_2 \oplus \xi.
\end{eqnarray*}
Obviously, we have $R_1+R_{2}+R_{3}+R_{4}=B^{\id}_1+B^{\id}_2$ and 
$( R_1+R_2)-(R_3+R_4) = \B^{\zeta_2}_1+\B^{\zeta_2}_2$. By applying the projection operator $\Theta_{\xi}$ we find that $\langle R
_1-R_2,R_3-R_4 \rangle_{\CC}$ accounts for a copy of the $\xi$ representation. We define 
\begin{eqnarray*}
 \B^{\xi}_1 = R_1-R_2, \qquad 
 \B^{\xi}_2 = R_3-R_4.
\end{eqnarray*}


Next, define the \emph{twisted vectors} as
\begin{eqnarray*}
 H_i & := & L_{ik}+L_{il}+L_{ji}, \\ V_i & := & L_{ki}+L_{li}+L_{ij},
\end{eqnarray*}
where $\{\{i,j\},\{k,l\}\}=\{\{1,2\},\{3,4\}\}$. 
For example, $V_2=L_{21}+L_{32}+L_{42}$ and $H_3=L_{31}+L_{32}+L_{43}$. 
%
The action $\rho_{\G}$ of ${\G}$ on these vectors  is given by 
$\sigma: V_i  \mapsto  V_{\sigma(i)}$ and 
$\sigma: H_i  \mapsto  H_{\sigma(i)}$,
again we have with the restriction of the defining representation of $\SG_{4}$ to ${\G}$ . 
As above, 
\begin{eqnarray*}
 \langle V_1,V_2,V_3,V_4 \rangle_{\CC}  \cong  \langle H_1,H_2,H_3,H_4 \rangle_{\CC} \cong \id\oplus \zeta_2 \oplus \xi.
\end{eqnarray*}
Notice that $\sum_{ i} H_{i}=\sum_{ i} V_{i}=\B^{\id}_1+\B^{\id}_2$,
$(H_{1}+H_{2})-(H_{3}+H_{4})=\B^{\zeta_2}_{1}+\B^{\zeta_2}
_{2}$ and 
$(V_{1}+V_{2})-(V_{3}+V_{4})=\B^{\zeta_2}_{1}-\B^{\zeta_2}
_{2}$.
By applying the projection operator $\Theta_{\xi}$ in $ \langle V_1,V_2,V_3,V_4 \rangle_{\CC} $ and $ \langle H_1,H_2,H_3,H_4 \rangle_{\CC}$, we find that 
$ \langle H_1-H_2,H_3-H_4  \rangle_{\CC}$ and $\langle V_1-V_2,V_3-V_4  \rangle_{\CC}$ account for the two other copies of $\xi$, so we define 
\begin{eqnarray*}
\begin{array}{lll}
 \B^{\xi}_3 =  H_1-H_2, &  \qquad  \qquad & 
  \B^{\xi}_5 = V_1-V_2,\\
 \B^{\xi}_4 = H_3-H_4, &  \qquad \qquad & 
 \B^{\xi}_6 = V_3-V_4.
\end{array}
\end{eqnarray*}


Putting all of these results together:
\begin{theorem}\label{specialbasis}
The Lie algebra $\mathfrak{L}_{GM}$ can be expressed as
\begin{eqnarray*}
\mathfrak{L}_{GM}& = & \langle\{L_{ij}\}_{1\leq i\neq j\leq 4}\rangle_{\CC}\\
& = & \langle \{\P_{\sigma}\}_{\sigma\in {\G}, \sigma\neq e} \cup \{\Q{12},\Q{34}\}\cup \{R_i\}_{1\leq i\leq 4}\cup\{H_j\}_{1\leq j\leq 4}\cup \{V_k\}_{1\leq k \leq 4}  \rangle_{\CC}, 
\end{eqnarray*}
with linear dependencies
\begin{eqnarray*}
\P_{(12)}+\P_{(34)}=\P_{(12)(34)},\\
\P_{(13)(24)}+\P_{(14)(23)} = \P_{(1324)}+\P_{(1423)}=\Q{12}+\Q{34}, \\
H_{1}+H_{2}=R_1+R_2=\Q{12}+\P_{(12)},\\
H_{3}+H_{4}=R_3+R_4=\Q{34}+\P_{(34)},\\
V_1+V_2=\Q{34}+\P_{(12)},\\
V_3+V_4=\Q{12}+\P_{(34)}.
\end{eqnarray*}
A basis for $\LL_{GM}$ consistent with the decomposition of Theorem \ref{dec_GM} is given by 
\begin{center}
\begin{tabular}{lll}
$ \B^{\id}_1 = \P_{(12)(34)}$, &  \qquad &$\B^{\xi}_1 =R_1-R_2$, \\
$\B^{\id}_2 = \P_{(13)(24)}+\P_{(14)(23)}$, & \qquad &  $\B^{\xi}_2 = R_3-R_4$, \\
$\B^{\sgn} =  \P_{(13)(24)}-\P_{(14)(23)}$, &  \qquad & $ \B^{\xi}_3 = H_1-H_2$, \\
$\B^{\zeta_1} =  \P_{(1324)}-\P_{(1423)}$,  &  \qquad &  $\B^{\xi}
_4 = H_3-H_4$, \\
$\B^{\zeta_2}_1 = \P_{(12)}-\P_{(34)}$,  & \qquad &   $\B^{\xi}_5 = V_1-V_2$, \\
$\B^{\zeta_2}_2 = \Q{12}-\Q{34}$,  &  \qquad &  $\B^{\xi}_6 = V_3-V_4$; \\
\end{tabular}
\end{center}
where $\lrc{ \B^{\xi}_1,\B^{\xi}_2} $, $\lrc{ \B^{\xi}_3,\B^{\xi}_4} $ and $\lrc{ \B^{\xi}_5,\B^{\xi}_6} $ are the three copies of $\xi$ in $\LL_{GM}$. 
With respect to this basis,  the Lie algebra structure of $\LL_{GM}$ is summarised in Table~\ref{tab:Liebracket}.
\end{theorem}


\section{The list of Lie Markov models with purine/pyrimidine symmetry}\label{sec5}

We proceed to give the list of Lie Markov models with  purine/pyrimidine symmetry, working up in dimension $d\leq 12$. 
For each $d$, Table \ref{tab:all_decs} lists all the possible decompositions allowed by Theorem \ref{dec_GM}. 
For each decomposition, all possible complex Lie subalgebras $\LL$ of $\LL
_{GM}$  are obtained by direct computation using code written by the authors and implemented in the open-source mathematical software SAGE \cite{SAGE}  (this code is available online at the website \cite{LMMweb}).
For each Lie algebra, we then impose that it has a stochastic basis  (see Definition  \ref{def:LMM}). 
%
%
Since a matrix  $\B^{\id}_{a,b}=a \B^{\id}_1+b \B^{\id}_2$, with $a,b>0$, has all its non-diagonal entries positive, the reader can notice that the above condition is guaranteed if $\LL$ contains such a matrix for, in this case, if $\{B_1,\ldots,B_t\}$ is a basis for $\LL$, then a stochastic basis for $\LL$ is given by $\{B_1+\lambda \B^{\id}_{a,b},\ldots, B_t+\lambda \B^{\id}_{a,b}\}$ as long as $\lambda>0$ is large enough.
%
%

For each model in the list, we describe a basis for the Lie algebra in terms of the vectors introduced in the Section~\ref{sec4} and  the rays of the stochastic cone arranged in orbits (see Table \ref{tab:ray_orbits}). 
Both data are required to completely describe the model. 
The general form of the stochastic rate matrix, as well as a permutation basis (a basis invariant under the action of $\G$), is also shown when it is not too  complicated. 
In particular, stochastic rate matrices are presented as linear combinations of the rays with non-negative coefficients. 
Since the rays are the generators of the stochastic cone, every stochastic rate matrix in the model can be expressed in this way (the reader should notice that in general, we cannot write down all the stochastic rate matrices of a model in terms of the same permutation basis if we require the coefficients to be non-negative).  
The name of each model has the form ``$d.r$'', where $d$ is the dimension of the model and $r$ is the number of rays of the corresponding stochastic cone (in particular, $d \leq r$). %
In case there is more than one model with a given dimension and number of rays, we will differentiate them by using letters: for example, $5.7a$, $5.7b$ and so on.

\begin{note}
Throughout the following list, we adopt the  notation $X_{ij}^{+}=X_i+X_j$ and $X_{ij}^{-}=X_i-X_j$, for $i,j\in \{1,2,3,4\}$ and $X\in \{R,H,V\}$.\qed
\end{note}

\paragraph*{Ray-orbits}
\quad \newline 
The rays of the stochastic cones of the forthcoming models appear in orbits of cardinality 1, 2, 4 and 8 (as is demanded by the orbit-stablizer theorem) that we call \emph{ray-orbits}. 
A system of generators for any of these ray-orbits is obtained as the $\G$-orbit of a rate matrix $Q$ in any of the rays of the family. Notice incidentally the action of $\G$ preserves the total sums of transition rates and of transversions rates of the rate matrices within the $\G$-orbit. 
%
%
%
%

For the Lie Markov models with symmetry $\G$, we explicitly describe the rays of the corresponding stochastic cone arranged in ray-orbits. 
In order to denote and compare these ray-orbits in a convenient fashion, we first  \emph{normalize} the generators of the rays, and take rate matrices whose trace is equal to $-1$ (recall the trace of the rate matrix can be understood as the expected number of changes  in one unit of time under the Markov process). Then, taking into account that the sum of transition rates and of transversion rates is constant, each ray-orbit is referred to 
as ``$\trip{r}{\frac{s}{s+v}}{\frac{v}{s+v}}$'', where 
\begin{description}
\item $r$ is the number of rays in the orbit: 1, 2, 4 or 8;
\item $s$ is the sum of the transition rates in (any matrix of) the orbit;

\item $v$ is the sum of the transversion rates in (any matrix of) the orbit.  
\end{description}
The reader is referred to Table \ref{tab:ray_orbits} for the whole list of ray-orbits arising in Lie Markov models with $\G$-symmetry.

\begin{note}
 From now on, we write 
 $\B^{\id}=\B^{\id}_{1}+\B^{\id}_{2}$, 
  $\B^{\zeta_2}=\B^{\zeta_2}_{1}+\B^{\zeta_2}_{2}$.\qed
\end{note}

\vspace*{4mm}  
\noindent
\textsc{Dimension One}

\vspace*{3mm}
\noindent
From Table~\ref{tab:Z2models} we see that there is only one abstract orbit of ${\G}$ with cardinality one, and it has decomposition $\id$.

\vspace*{2mm}

\subsection*{$\id$}
The general Markov model contains two copies of the trivial representation, so we can consider the subspace generated by any linear combination 
$a \B^{\id}_1+b \B^{\id}_2$. Moreover, since $ [\B^{\id}_1,\B^{\id}_2]=0$, we see the subspace generated by any $a \B^{\id}_1+b \B^{\id}_2$, is a Lie algebra for any fixed choice $a,b\in \CC$. When we request these spaces to have a stochastic basis, we have to restrict to the condition $a,b\geq 0$. Therefore, we conclude:
\begin{theorem}\label{1_dim}
In the 4-state case, there is a continuum of one-dimensional Lie Markov models with ${\G}$ symmetry and decomposition $\id$. Each model in the family has the form
\begin{eqnarray*}
\LL=\langle B^{\id}_{a,b} \rangle_{\CC} \quad a,b\geq 0,
\end{eqnarray*}
where $a+b=1, a,b\geq 0$ and 
\begin{eqnarray}
B^{\id}_{a,b}:=a \B^{\id}_1+b\B^{\id}_2=
\left(
\begin{array}{cccc}
	*  & a & b & b \\
	a & * & b & b\\
	b & b & * & a  \\
	b & b & a & *
\end{array}\right).\nonumber
\end{eqnarray}
where we use $\ast$ to indicate the diagonal entry needed for the column to sum to zero. 
\end{theorem}


\begin{remark}\label{no_continuum} 
This result is not completely satisfactory as all these models will appear as 1-dimensional Lie subalgebras  of the 2-dimensional Lie Markov model
$ \langle \B^{\id}_1,\B^{\id}_2 \rangle_{\CC}$.
This situation is quite general and we will avoid the consequent redundancy in the present list by considering families of Lie Markov models depending on some parameters as submodels of a  Lie Markov models with larger dimension. Then, the family of models in Theorem  \ref{1_dim} should be regarded as a Lie Markov model with decomposition $2 \id$.  
%

On the other hand, notice that if we expand the symmetry  and request the models in the family of Theorem  \ref{1_dim} to have the symmetry of $\SG
_4$, we are lead to the constraint  $a=b$, which corresponds to the Jukes-Cantor model \cite{JC69}. Of course, this model already appeared as a Lie Markov model with symmetry $\SG_4$ in \cite{LMM}.\qed
\end{remark}

\noindent
\framebox{\emph{Model 1.1}} Take  $\mathfrak{L}=\lrc{\B^{\id}}$. The stochastic cone has only one ray, spanned by $\B^{\id}$. Therefore, in this case, we only have a ray-orbit. 
We refer to it by the ray-orbit $\trip{1}{\frac{1}{3}}{\frac{2}{3}}$ (see Table \ref{tab:ray_orbits}). 
The generic stochastic rate matrix is 
\begin{eqnarray*}
\left(\begin{array}{cccc}
* & 1 & 1 & 1 \\
1 &* & 1 & 1 \\
1 & 1 & * & 1 \\
1 & 1 & 1 & *
\end{array}\right).
\end{eqnarray*}

%

\vspace*{4mm}  
\noindent
\textsc{Dimension Two}

\vspace*{1mm}
  \subsection*{$\id \oplus \sgn$} We have 
 $ [\B^{\id}_1,\B^{\sgn}]=[\B^{\id}_2,\B^{\sgn}]=0$, so for any fixed $a,b\geq 0$ with $a+b=1$ and $b\neq 0$, there is a well-defined Lie Markov model: 
 \begin{eqnarray*}
 \LL=\lrc{ B^{\id}_{a,b},\B^{\sgn}}\cong \id\oplus \sgn.
 \end{eqnarray*}
 The condition $b\neq 0$ is needed to ensure that the dimension of the stochastic cone is equal to the dimension of the Lie algebra.   
As in Remark~\ref{no_continuum}, these models are  considered as submodels of the model with decomposition $2\id\oplus\sgn$ (see Model 3.3a).

\subsection*{$\id  \oplus  \zeta_1$} Since $[\B^{\id}_1,\B^{\zeta_1}]=[\B^{\id}_2,\B^{\zeta_1}]= 0$, we find that, for any choice of $a,b\geq 0$ with $a+b=1$ and $b\neq 0$,  
 \begin{eqnarray*}
  \LL=\lrc{ B^{\id}_{a,b},\B^{\zeta_1} }\cong \id\oplus \zeta_1,
 \end{eqnarray*}
 provides a 2-dimensional Lie Markov model. 
 These are submodels of the 3-dimension model with decomposition $2\id +\zeta_1$ (see Model 3.3b).



\subsection*{$\id\oplus \zeta_2$} We find the same situation for this decomposition. The following Lie algebras will appear as submodels of the 3-dimensional model indicated: 
\begin{description}
 \item $\LL=\lrc{ B^{\id}_{a,b}, \B^{\zeta_2}}$ is a submodel of Model 3.4; 
  \item $\LL=\lrc{ B^{\id}_{a,b},\B^{\zeta_2}_{1}}$ is a submodel of Model 3.3c. 
 \end{description}

As a special case of the last family of models, when we take $a=1$ and $b
=0$, we obtain the following model:
 \quad
\newline

\noindent
 \framebox{\emph{Model 2.2a}} $\LL=\lrc{ \B^{\id}_1,\B^{\zeta_2}_1}$. 
 The stochastic cone has two rays generated by $\P_{(12)}$ and  $\P_{(34)}$, which form the ray-orbit $\trip{2}{1}{0}$ of Table \ref{tab:ray_orbits}.
The general stochastic rate matrix is 
\begin{eqnarray*}
\left(\begin{array}{cccc}
* & \alpha & 0 & 0 \\
\alpha &* & 0 & 0 \\
0 & 0 & * & \beta \\
0 & 0 & \beta & *
\end{array}\right),\quad \alpha, \beta \geq 0.
\end{eqnarray*}
This model gives a reducible Markov chain, that is, it is not possible to get to some states from some other states. We see that the purine states $A$ and $G$ communicate with each other, and the same for the pyrimidine states $C$ and $T$ (transitions) while no replacement between purines are pyrimidines  (transversions) is allowed.  

Apart from these models, our analysis of 1-dimensional Lie Markov  models produces another 2-dimensional model with decomposition $2 \id$. 

\vspace*{2mm}
\subsection*{$2 \id$} Of course, there is only one possible model with this decomposition. Namely, 

\vspace*{3mm}
 \noindent
 \framebox{\emph{Model 2.2b}} $\LL=\lrc{ \B^{\id}_1,\B^{\id}_2}$. 
If we focus on the stochastic rate matries, we find a cone with 2 rays, corresponding to the (see Table~\ref{tab:ray_orbits}): 
ray-orbit $\trip{1}{1}{0}=\{L_{(12)(34)}\}$, and ray-orbit $\trip{1}{0}{1}
=\{L_{(13)(24)}+L_{(14)(23)}\}$. 
The Lie algebra is abelian and the stochastic rate matrices for this model are given by 
\begin{eqnarray*}
 Q=\left(\begin{array}{cccc}
* & \alpha & \beta & \beta\\
\alpha & * & \beta & \beta \\
\beta & \beta & * & \alpha \\
\beta& \beta & \alpha & *
\end{array}\right), \quad \alpha, \beta \geq 0.
\end{eqnarray*}
\textit{Permutation basis}: $\P_{(12)(34)},\P_{(13)(24)}+\P_{(14)(23)}.$
 
 This model corresponds to the Kimura model with 2 parameters \cite{kimura1980}.


\vspace*{4mm}  
\noindent
\textsc{Dimension Three}

\vspace*{2mm}
 \subsection*{$2\id \oplus  \sgn$} 
There is only one model with this decomposition: 
 
 \vspace*{3mm}
 \noindent
 \framebox{\emph{Model  3.3a}}  $ \LL=\lrc{ \B^{\id}_1,\B^{\id}_2,\B^{\sgn}}$  
is an abelian Lie Markov model. 
The stochastic cone has 3 rays in 2 ray-orbits: ray-orbit $\trip{1}{1}{0}=\{L_{(12)(34)}\}$, and ray-orbit $\trip{2}{0}{1}a=\{L_{(13)(24)},L_{(14)(23)}\}$
(see Table \ref{tab:ray_orbits}). 
The general stochastic rate matrix is 
\begin{eqnarray*}
\left(\begin{array}{cccc}
* & \alpha & \beta & \gamma \\
\alpha & * & \gamma & \beta \\
\beta & \gamma &*& \alpha \\
\gamma & \beta & \alpha & *
\end{array}\right), \quad \alpha, \beta, \gamma \geq 0.
\end{eqnarray*}
%
\textit{Permutation basis}: $\P_{(12)(34)},\P_{(13)(24)},\P_{(14)(23)}.$

Of course, this is the Kimura model with 3 parameters  \cite{kimura1981}.  Note this is the group-based model corresponding to $\mathbb{Z}_2\times \mathbb{Z}_2\cong \left\{e,(12)(34),(13)(24),(14)(23)\right\}$. 

\vspace*{2mm}
\subsection*{$2\id  \oplus  \zeta_1$}
There is only one model with this decomposition:

 \vspace*{3mm}
 \noindent  \framebox{\emph{Model 3.3b}}
$ \LL=\lrc{ \B^{\id}_1,\B^{\id}_2,\B^{\zeta_1} }$
is a 3-dimensional abelian Lie Markov model. The stochastic cone has 3 rays, in 2 ray-orbits: ray-orbit $\trip{1}{1}{0}=\{L_{(12)(34)}\}$, and ray-orbit $\trip{2}{0}{1}b=\{L_{(1324)},L_{(1423)}\}$. 
The general stochastic rate matrix is 
\begin{eqnarray*}
\left(\begin{array}{cccc}
* & \alpha & \beta & \gamma \\
\alpha &* & \gamma & \beta \\
\gamma & \beta &* & \alpha \\
\beta & \gamma & \alpha & *
\end{array}\right), \quad \alpha, \beta, \gamma \geq 0.
\end{eqnarray*}
This new model may be regarded as a ``twisted'' version of the Kimura model with three parameters. 

\textit{Permutation basis}: $\P_{(12)(34)},\P_{(1324)},\P_{(1423)}.$

Note this is the group-based model corresponding to $\mathbb{Z}_4\cong \left\{e,(1324),(12)(34),(1423)\right\}$.

\vspace*{2mm}
\subsection*{$2\id \oplus  \zeta_2$} There are two models with this decomposition: 
 
 \vspace*{3mm}
 \noindent
 \framebox{\emph{Model 3.3c}} $ \LL=\lrc{ \B^{\id}_1,\B^{\id}_2,\B^{\zeta_2}_1 }$ is a 3-dimensional abelian Lie algebra. The stochastic cone has 3 rays, in 2 ray-orbits: 
ray-orbit $\trip{1}{0}{1}=\{L_{(13)(24)}+L_{(14)(23)}\}$, and ray-orbit $\trip{2}{1}{0}=\{L_{(12)},L_{(34)}\}$. 
The general stochastic rate matrix is 
\begin{eqnarray*}
\left(\begin{array}{cccc}
*& \alpha & \beta & \beta \\
\alpha &* & \beta & \beta \\
\beta & \beta & * & \gamma \\
\beta & \beta & \gamma &*
\end{array}\right), \quad \alpha, \beta, \gamma \geq 0.
\end{eqnarray*}
%
%
\textit{Permutation basis}: $\P_{(12)},\P_{(34)},\P_{(1324)}+\P_{(1423)}.$

\vspace*{3mm}

\noindent
 \framebox{\emph{Model 3.4}} $ \LL=\lrc{ \B^{\id}_1,\B^{\id}_2,\B^{\zeta_2}}$ is a 3-dimensional Lie algebra. The stochastic cone has 4 rays, in 3 ray-orbits:  
ray-orbit $\trip{1}{0}{1} = \{\P_{(13)(24)}+\P_{(14)(23)}\}$, 
ray-orbit $\trip{1}{1}{0} = \{\P_{(12)(34)}\}$, and
ray-orbit $\trip{2}{1/3}{2/3} = \{R_{12}^+,R_{34}^+\}$.
This is the first model with $\G$ symmetry where the number of rays is larger than the dimension of the model. %
It is also the first case where the Lie algebra $\LL$ is not abelian: the Lie algebra structure is given by 
\begin{eqnarray*}
 & {[} \P_{(13)(24)}+\P_{(14)(23)},R_{ij}^{+} ]=R_{kl}^{+}-R_{ij}^{+},   \qquad  & {[} \P_{(12)(34)},R_{ij}^{+} ]=0, \\
 & {[} \P_{(13)(24)}+\P_{(14)(23)},\P_{(12)(34)} ]=R_{kl}^{+},  \qquad &  
{[} R_{ij}^{+}, R_{kl}^{+} ]=R_{ij}^{+}-R_{kl}^{+},
\end{eqnarray*}
for $\{ij,kl\}=\{12,34\}$. 
The general stochastic rate matrix is 
\begin{eqnarray*}
\left(\begin{array}{cccc}
* & \alpha + \gamma & \beta + \gamma & \beta + \gamma \\
\alpha + \gamma &*& \beta + \gamma & \beta + \gamma \\
\beta + \delta & \beta + \delta & * & \alpha + \delta \\
\beta + \delta & \beta +\delta & \alpha +\delta &* 
\end{array}\right), \quad \alpha, \beta, \gamma, \delta \geq 0.
\end{eqnarray*}
\textit{Permutation basis}: $\P_{(12)(34)},R_{12}^{+},R_{34}^{+}.$


\vspace*{4mm}  
\noindent 
\textsc{Dimension Four}

\vspace*{3mm}
\noindent
 Lie Markov models with this dimension appear as special submodels of forthcoming models $5.7a$, $5.7b$ and $5.7c$ with decomposition $2\id \oplus  \sgn \oplus \xi$, when we restrict the identity component of their Lie algebra $\LL$ to a subspace $\lrc{ \B^{\id}_{a,b}} $ with $a,b\geq 0$. 
 The reader can check that depending on the values of $a$ and $b$ the number of rays of the cones of these models may vary. 

 \vspace*{2mm}
 \subsection*{$\id  \oplus  \sgn  \oplus  \zeta_1  \oplus  \zeta_2$}  The models with this decomposition appear as special cases of Model 5.6a with decomposition $2 \id  \oplus  \sgn  \oplus  \zeta_1  \oplus  \zeta_2$ (see Remark \ref{no_continuum}).
 
 \vspace*{2mm}
 \subsection*{$\id  \oplus  \zeta_2  \oplus  \xi$} Similarly, these models  are special cases of the models $5.6b$, $5.11a$, $5.11b$, $5.11c$ and $5.16$ with decomposition $2 \id  \oplus  \zeta_2  \oplus  \xi$.
As a particular case, if we request these models to have $\SG_4$ symmetry, we obtain the restriction $a=b$, leading to the Felsenstein 1981 model \cite{felsenstein1981}:

\vspace*{3mm}

\noindent
 \framebox{\emph{Model 4.4a}}  $ \LL=\lrc{ \B^{\id},\B^{\zeta_2},
\B^{\xi}_1,\B^{\xi}_2 }$ is a 4-dimensional Lie algebra. The stochastic cone has 4 rays  in one single ray-orbit: $\trip{4}{\frac{1}{3}}{\frac{2}{3}}a=\{R_1,R_2,R_3,R_4\}$, and the Lie algebra structure is given by $ [R_i,R_j]=R_i-R_j$.
The general stochastic rate matrix is
\begin{eqnarray*}
\left(\begin{array}{cccc}
* & \alpha  & \alpha & \alpha \\
\beta & *& \beta & \beta \\
 \gamma & \gamma & * & \gamma \\
\delta & \delta & \delta & *
\end{array}\right), \quad \alpha, \beta, \gamma, \delta \geq 0.
\end{eqnarray*}
Of course, this is Felsenstein 1981 model (see \cite{felsenstein1981}).

\textit{Permutation basis}: $R_{1},R_{2},R_{3},R_{4}.$


\vspace*{2mm}
\subsection*{$2 \id  \oplus  2 \zeta_2$} There is only one model with this decomposition. 
 
 \vspace*{3mm}
 
 \noindent
 \framebox{\emph{Model 4.4b}}  Take $ \LL=\lrc{ \B^{\id}_1,\B^{\id}_2,\B^{\zeta_2}_1,\B^{\zeta_2}_2 }$.  The stochastic cone has 4 rays, in 2 ray-orbits: 
ray-orbit $\trip{2}{0}{1}c = \{\Q{12},\Q{34}\}$, and
ray-orbit $\trip{2}{1}{0} =\{\P_{(12)},\P_{(34)}\}$.
The Lie algebra is given by 
\begin{eqnarray*}
{[}\P_{(12)},\P_{(34)}]& = & 0, \\ {[}\P_{(12)},\Q{ij}]& = &{[}\P
_{(34)},\Q{ij}]=0, \quad ij\in \{12,34\} \\
{[}\Q{12},\Q{34}]& = &2(\Q{34}-\Q{12})+2(\P_{(34)}-\P_{(12)}).
\end{eqnarray*}
The general stochastic rate matrix is
\begin{eqnarray*}
\left(\begin{array}{cccc}
* & \alpha  & \beta & \beta \\
\alpha & *& \beta & \beta \\
\gamma & \gamma & * & \delta \\
\gamma & \gamma & \delta & *
\end{array}\right), \quad \alpha, \beta, \gamma, \delta \geq 0
\end{eqnarray*}
%
%
%
\textit{Permutation basis}: $\P_{(12)},\P_{(34)},\Q{12},\Q{34}.$

\vspace*{2mm}
\subsection*{$2 \id  \oplus  \sgn  \oplus  \zeta_2$}
There is only one model with this decomposition. 
 
 \vspace*{3mm}
 \noindent
 \framebox{\emph{Model 4.5a}} $ \LL=\lrc{ \B^{\id}_1,\B^{\id}_2,\B^{\sgn},\B^{\zeta_2} }$ is a 4-dimensional Lie algebra. 
The stochastic cone has 5 rays spanned, in 3 ray-orbits: $\trip{1}{1}{0}$, $\trip{2}{0}{1}a$ and $\trip{2}{\frac{1}{3}}{\frac{2}{3}}$.
The general stochastic rate matrix is
\begin{eqnarray*}
\left(\begin{array}{cccc}
* & \alpha + \delta & \beta + \delta & \gamma + \delta \\
\alpha + \delta & * & \gamma + \delta& \beta + \delta \\
\beta + \varepsilon & \gamma + \varepsilon& * & \alpha + \varepsilon \\
\gamma + \varepsilon& \beta + \varepsilon & \alpha + \varepsilon & *
\end{array}\right), \quad \alpha, \beta, \gamma, \delta, \varepsilon \geq 0.
\end{eqnarray*}
%
\textit{Permutation basis}: $R_{12}^{+},R_{34}^{+},\P_{(13)(24)},\P_{(14)(23)}.$

\vspace*{2mm}
\subsection*{$2 \id  \oplus  \zeta_1  \oplus  \zeta_2$} 
There is only one model with this decomposition. 

\vspace*{3mm}
\noindent
 \framebox{\emph{Model 4.5b}} $ \LL=\lrc{ \B^{\id}_1,\B^{\id}_2,\B^{\zeta_1},\B^{\zeta_2} }$ is a 4-dimensional Lie algebra. The stochastic cone has 5 rays, in 3 ray-orbits: $\trip{1}{1}{0}$, $\trip{2}{0}{1}b$ and $\trip{2}{\frac{1}{3}}{\frac{2}{3}}$. 
 Then, 
the general stochastic rate matrix is
\begin{eqnarray*}
\left(\begin{array}{cccc}
* & \alpha + \delta & \beta + \delta & \gamma + \delta \\
\alpha + \delta & * & \gamma + \delta& \beta + \delta \\
\gamma + \varepsilon & \beta + \varepsilon& * & \alpha + \varepsilon \\
\beta + \varepsilon& \gamma + \varepsilon & \alpha + \varepsilon & *
\end{array}\right), \quad \alpha, \beta, \gamma, \delta, \varepsilon  \geq 0.
\end{eqnarray*}
%
%
\textit{Permutation basis}: $R_{12}^{+},R_{34}^{+},\P_{(1324)},\P_{(1423)}.$

\vspace*{3mm}

The remaining models, with dimensions 5 to 12, are presented in the Table \ref{tab:models}, with a complete list in explicit form available at \cite{LMMweb}.
The first column of the table gives the name of the model and the second column gives a basis for the corresponding Lie subalgebra. The third column gives the ray-orbits for the stochastic cone of that Lie subalgebra. 

\begin{remark}
A number of models deserve some remarks. 
\begin{itemize}
 \item  Model 5.6b can be regarded as the vector sum of the Felsenstein 1981 model \cite{felsenstein1981} and the Kimura model with 2 parameters.

\item  Model 6.7a already appeared in \cite{LMM} under the name $K3ST+F81$ and it has $\SG_4$ symmetry. A permutation basis for this model is $\{\P
_{(13)(24)}, \P_{(14)(23)}$, 
$R_1, R_2, R_3, R_4\}$, which is not invariant under the action of $\SG_4$.
As its name suggests, this model is the vector sum of the Kimura 3 parameter and Felsenstein 1981 models.
A permutation basis consistent with the symmetry $\SG_4$ is given by the vectors $W_{ij}$ of Example \ref{Ex:F81+K3ST} (see also  \cite{LMM}).

\item Of course, Model 12.12 is the general Markov model and we include it in the list for completion. \qed
 \end{itemize}
 
\end{remark}

\begin{table}

 \caption{\label{tab:models} \small{List of Lie Markov models with purine/pyrimidine symmetry  for dimension 5, 6, 8, 9, 10 and 12 (the Lie Markov models with dimension 1 to 4 are described within the text). The first column gives the name of the model, while the second and the third column provide a basis of the corresponding Lie subalgebra and the ray-orbits of the corresponding stochastic cone, respectively (see Table \ref{tab:ray_orbits}). }}
 \vspace*{2mm}
 \begin{adjustwidth}{-1cm}{-1cm}
 \begin{tabular}{|c|p{6.8cm}|p{8.2cm}|}
\hline
\multicolumn{3}{|l|}{$2 \id  \oplus  \sgn  \oplus  \xi$} \\ \hline
\textbf{5.7a} & $\lrc{ \B^{\id}_1,\B^{\id}_2,\B^{\sgn},  \B^{\xi}_1,\B^{\xi}_2}$  & $\trip{1}{1}{0}, \trip{2}{0}{1}a, \trip{4}{\frac{1}{3}}{\frac{2}{3}}d$ \\ 
\textbf{5.7b} & $\lrc{ \B^{\id}_1,\B^{\id}_2,\B^{\sgn},
\B^{\xi}_3,\B^{\xi}_4}$  & $\trip{1}{1}{0},\trip{2}{0}{1}a,\trip{4}{\frac{1}{3}}{\frac{2}{3}}e$ \\ 
\textbf{5.7c} & $\lrc{\B^{\id}_1,\B^{\id}_2,\B^{\sgn}, \B^{\xi}_5,\B^{\xi}_6}$  & $\trip{1}{1}{0},\trip{2}{0}{1}a,\trip{4}{\frac{1}{3}}{\frac{2}{3}}f$ \\\hline \hline 
 \multicolumn{3}{|l|}{$2 \id  \oplus  \sgn \oplus  \zeta_1  \oplus  \zeta
_2$} \\\hline 
 \textbf{5.6a} & $\lrc{ \B^{\id}_1,\B^{\id}_2,\B^{\sgn},\B^{\zeta_1},\B^{\zeta_2}_1}$  & $\trip{2}{0}{1}a, \trip{2}{0}{1}b, \trip{2}{1}{0}$ \\ \hline \hline     
 \multicolumn{3}{|l|}{$2 \id  \oplus  \zeta_2  \oplus  \xi$} \\\hline 
  \textbf{5.6b} & $\lrc{\B^{\id}_1,\B^{\id}_2,\B^{\zeta_2},\B^{\xi}_1,\B^{\xi}_2}$ & $\trip{1}{0}{1},\trip{1}{1}{0}, \trip{4}{\frac{1}{3}}{\frac{2}{3}}a$ \\
\textbf{5.16} & $\lrc{ \B^{\id}_1,\B^{\id}_2,\B^{\zeta_2},
\B^{\xi}_5,\B^{\xi}_6}$ & $\trip{1}{0}{1},\trip{1}{1}{0},\trip{2}{\frac{1}{3}}{\frac{2}{3}},  \trip{4}{\frac{1}{7}}{\frac{6}{7}},\trip{4}{\frac{1}{3}}{\frac{2}{3}}f,\trip{4}{\frac{3}{5}}{\frac{2}{5}}$ \\
\textbf{5.11a} & $\lrc{ \B^{\id}_1,\B^{\id}_2,\B^{\zeta_2}_1,
\B^{\xi}_1,\B^{\xi}_2}$ &
 $\trip{1}{0}{1}, \trip{2}{1}{0}, \trip{4}{\frac{1}{3}}{\frac{2}{3}}d, \trip{4}{\frac{1}{5}}{\frac{4}{5}}a$ \\
\textbf{5.11b} & $\lrc{ \B^{\id}_1,\B^{\id}_2,\B^{\zeta_2}_1,
\B^{\xi}_3,\B^{\xi}_4}$ & 
  $\trip{1}{0}{1},\trip{2}{1}{0},\trip{4}{\frac{1}{3}}{\frac{2}{3}}e,\trip{4}{\frac{1}{5}}{\frac{4}{5}}b$ \\
\textbf{5.11c} & $\lrc{ \B^{\id}_1,\B^{\id}_2,\B^{\zeta_2}_1,
\B^{\xi}_5,\B^{\xi}_6}$  & 
$\trip{1}{0}{1},\trip{2}{1}{0},\trip{4}{\frac{1}{3}}{\frac{2}{3}}f,\trip{4}{\frac{1}{5}}{\frac{4}{5}}c$\\\hline \hline
\multicolumn{3}{|l|}{$2 \id  \oplus  \sgn \oplus  \zeta_1  \oplus  2 \zeta_2$}\\ \hline 
  \textbf{6.6} & $\lrc{ \B^{\id}_1,\B^{\id}_2,\B^{\sgn},\B^{\zeta_1},\B^{\zeta_2}_1,\B^{\zeta_2}_2 }$ &  $\trip{2}{1}{0},\trip{4}{0}{1}e$ \\ \hline \hline
\multicolumn{3}{|l|}{$2 \id  \oplus  \sgn  \oplus  \zeta_2  \oplus  \xi$}\\ \hline 
  \textbf{6.7a} & $\lrc{ \B^{\id}_1,\B^{\id}_2,\B^{\sgn},\B^{\zeta_2}, \B^{\xi}_1,\B^{\xi}_2}$ & 
$\trip{1}{1}{0},\trip{2}{0}{1}a,\trip{4}{\frac{1}{3}}{\frac{2}{3}}a$ \\
\textbf{6.17a} & $\lrc{ \B^{\id}_1,\B^{\id}_2,\B^{\sgn},\B^{\zeta_2},\B^{\xi}_5,\B^{\xi}_6}$ & 
 $\trip{1}{1}{0},\trip{2}{0}{1}a,\trip{2}{\frac{1}{3}}{\frac{2}{3}},\trip{4}{\frac{1}{7}}{\frac{6}{7}},\trip{4}{\frac{1}{3}}{\frac{2}{3}}f,\trip{4}{\frac{3}{5}}{\frac{2}{5}}$ \\\hline \hline
\multicolumn{3}{|l|}{$2 \id  \oplus  \zeta_1  \oplus  \zeta_2  \oplus  \xi$}\\ \hline 
 \textbf{6.7b} & $\lrc{ \B^{\id}_1,\B^{\id}_2,\B^{\zeta_1},\B^{\zeta
_2}, \B^{\xi}_1,\B^{\xi}_2}$ &   
 $\trip{1}{0}{1},\trip{2}{0}{1}b,\trip{4}{\frac{1}{3}}{\frac{2}{3}}a$\\
\textbf{6.17b} & $\lrc{ \B^{\id}_1,\B^{\id}_2,\B^{\zeta_1},\B^{\zeta
_2},\B^{\xi}_5,\B^{\xi}_6}$ & 
$\trip{1}{1}{0},\trip{2}{0}{1}b,\trip{2}{\frac{1}{3}}{\frac{2}{3}},\trip{4}{\frac{1}{7}}{\frac{6}{7}},\trip{4}{\frac{1}{3}}{\frac{2}{3}}f, \trip{4}{\frac{3}{5}}{\frac{2}{5}}$\\\hline \hline
\multicolumn{3}{|l|}{$2 \id  \oplus  2 \zeta_2  \oplus  \xi$}\\ \hline 
 \textbf{6.8a} & $\lrc{ \B^{\id}_1,\B^{\id}_2,\B^{\zeta_2}_1,\B^{\zeta_2}_2, \B^{\xi}_1,\B^{\xi}_2}$ & 
 $\trip{2}{0}{1}c,\trip{2}{1}{0},\trip{4}{\frac{1}{3}}{\frac{2}{3}}a$ \\
\textbf{6.8b} & $\lrc{ \B^{\id}_1,\B^{\id}_2,\B^{\zeta_2}_1,\B^{\zeta_2}_2,\B^{\xi}_5,\B^{\xi}_6}$ & 
 $\trip{2}{0}{1}c,\trip{2}{1}{0},\trip{4}{\frac{1}{3}}{\frac{2}{3}}b$ \\\hline \hline
\multicolumn{3}{|l|}{$2 \id  \oplus  2 \zeta_2  \oplus  2 \xi$}\\ \hline

\textbf{8.8} & $\lrc{ \B^{\id}_1,\B^{\id}_2,\B^{\zeta_2}_1,\B^{\zeta_2}_2,  \B^{\xi}_1,\B^{\xi}_2,\B^{\xi}_3,\B^{\xi}_4}$ &  
 $\trip{4}{0}{1}c,\trip{4}{1}{0}a$ \\
\textbf{8.16} & $\lrc{ \B^{\id}_1,\B^{\id}_2,\B^{\zeta_2}_1,\B^{\zeta_2}_2, \B^{\xi}_1,\B^{\xi}_2,\B^{\xi}_5,\B^{\xi}_6 }$ & 
 $\trip{2}{0}{1}c,\trip{2}{1}{0},\trip{4}{0}{1}a,\trip{4}{\frac{1}{3}}{\frac{2}{3}}a,\trip{4}{\frac{1}{3}}{\frac{2}{3}}b$\\\hline \hline
\multicolumn{3}{|l|}{$2 \id  \oplus  \sgn \oplus  \zeta_1  \oplus  2 \zeta_2 \oplus \xi$}\\ \hline 
 \textbf{8.10a} & $\lrc{ \B^{\id}_1,\B^{\id}_2,\B^{\sgn},\B^{\zeta_1},\B^{\zeta_2}_1,\B^{\zeta_2}_2,\B^{\xi}_1, \B^{\xi}_2}$ &  $\mathrm{\trip{2}{1}{0},\trip{4}{0}{1}e,\trip{4}{\frac{1}{3}}{\frac{2}{3}}a}$\\
\textbf{8.10b} &  $\lrc{ \B^{\id}_1,\B^{\id}_2,\B^{\sgn},\B^{\zeta_1},\B^{\zeta_2}_1,\B^{\zeta_2}_2,\B^{\xi}_{5},\B^{\xi}_{6}} $ &  $\trip{2}{1}{0},\trip{4}{0}{1}e,\trip{4}{\frac{1}{3}}{\frac{2}{3}}b$\\\hline \hline
\multicolumn{3}{|l|}{$2 \id  \oplus  \sgn  \oplus  \zeta_2  \oplus  2 \xi$}\\ 
 \hline
\textbf{8.17} & $\lrc{ \B^{\id}_1,\B^{\id}_2,\B^{\sgn},\B^{\zeta_2}, \B^{\xi}_1,\B^{\xi}_2,\B^{\xi}_3,\B^{\xi}_4}$ & 
 $\trip{1}{1}{0},\trip{4}{0}{1}d,\trip{4}{\frac{1}{3}}{\frac{2}{3}}a,\trip{4}{\frac{1}{2}}{\frac{1}{2}}a,\trip{4}{\frac{3}{5}}{\frac{2}{5}}$\\
 \textbf{8.18} & $\lrc{ \B^{\id}_1,\B^{\id}_2,\B^{\sgn},\B^{\zeta_2},\B^{\xi}_1,\B^{\xi}_2, \B^{\xi}_3,\B^{\xi}_4}$ & $\trip{2}{0}{1}a,\trip{4}{0}{1}b,\trip{4}{\frac{1}{3}}{\frac{2}{3}}a,\trip{4}{\frac{1}{3}}{\frac{2}{3}}b, \trip{4}{1}{0}b$\\\hline \hline
 \multicolumn{3}{|l}{$2 \id \oplus  \sgn  \oplus  \zeta_1  \oplus  \zeta
_2 \oplus  2 \xi$}\\ \hline 
\textbf{9.20a}  & $\lrc{ \B^{\id}_1,\B^{\id}_2,\B^{\sgn},\B^{\zeta_1}, \B^{\zeta_2}_1,  \B^{\xi}_1,\B^{\xi}_2,\B^{\xi}_3,\B^{\xi}_4}$ &  $\mathrm{\trip{4}{1}{0}a,\trip{2}{0}{1}a,\trip{2}{0}{1}b,\trip{4}{0}{1}b,\trip{8}{0}{1}b}$\\
 \textbf{9.20b} & $\lrc{ \B^{\id}_1,\B^{\id}_2,\B^{\sgn},\B^{\zeta_1},\B^{\zeta_2}_2,
  \B^{\xi}_3,\B^{\xi}_4,\B^{\xi}_5,\B^{\xi}_6}$
 & $\trip{2}{0}{1}b,\trip{2}{1}{0},\trip{4}{0}{1}f,\trip{4}{\frac{1}{2}}{\frac{1}{2}}b,\trip{8}{\frac{1}{3}}{\frac{2}{3}}b$ \\\hline \hline
 \multicolumn{3}{|l|}{$2 \id \oplus  \sgn  \oplus  \zeta_1  \oplus  2 \zeta_2 \oplus  2 \xi$}\\ \hline 
  \textbf{10.12} & $\lrc{ \B^{\id}_1,\B^{\id}_2,\B^{\sgn}, \B^{\zeta_1},\B^{\zeta_2}_1,\B^{\zeta_2}_2,
 \B^{\xi}_1,\B^{\xi}_2,\B^{\xi}_3,\B^{\xi}_4}$ &  
 $\trip{4}{0}{1}c,\trip{4}{0}{1}e,\trip{4}{1}{0}a$ \\
  \textbf{10.34} & $\lrc{ \B^{\id}_1,\B^{\id}_2,\B^{\sgn},\B^{\zeta_1}, \B^{\zeta_2}_1,\B^{\zeta_2}_2,
 \B^{\xi}_1,\B^{\xi}_2,\B^{\xi}_5,\B^{\xi}_6}$ &  
 $\trip{2}{1}{0},\trip{4}{0}{1}d,\trip{4}{0}{1}e,\trip{4}{\frac{1}{3}}{\frac{2}{3}}a,\trip{4}{\frac{1}{3}}{\frac{2}{3}}b,\ldots \newline \hspace*{4.7cm} \ldots \trip{8}{\frac{1}{3}}{\frac{2}{3}}a, \trip{8}{1}{1}$\\\hline \hline
 \multicolumn{3}{|l|}{$2 \id \oplus  \sgn  \oplus  \zeta_1  \oplus  2 \zeta_2 \oplus  3 \xi$}\\ \hline 
 \textbf{12.12} & $\LL_{GM}$ & $\trip{4}{1}{0}a,\trip{8}{0}{1}a$\\ \hline
 \end{tabular}
 \end{adjustwidth}
 \thisfloatpagestyle{empty}
 \end{table}

\begin{remark}\label{HKY}
 The reader may notice the resemblance of Model 5.6b wih the HKY model (\cite{HKY}):
 \begin{eqnarray*}
Q_{5.6b}=\left(\begin{array}{cccc}
* & a+x & b+ x & b + x \\
a +y & * &b +y & b +y \\
b + z & b + z &* & a + z \\
b +t & b +t & a +t & *
\end{array}\right)    \qquad Q_{HKY}=\left(\begin{array}{cccc}
* & \pi_A \alpha & \pi_A \beta & \pi_A \beta\\
\pi_G \alpha & * & \pi_G \beta& \pi_G \beta\\
 \pi_C \beta & \pi_C \beta&* & \pi_C \alpha \\
\pi_T \beta& \pi_T \beta & \pi_T \alpha & *
\end{array}\right), 
 \end{eqnarray*}
 where $\pi_A+\pi_C+\pi_G+\pi_T=1$, and all these parameters are non-negative. 
Although the rates of these models depend on the parameters in a different way, the rate-matrices $Q_{HKY}$ and $Q_{5.6b}$ have the same structure. It is interesting to notice that the form of the non-diagonal entries of $Q_{5.6b}$ arises from the corresponding entries in $Q_{HKY}$ just by applying minus the logarithm, producing the following correspondance between the parameters of both models
\begin{eqnarray*}
 x=-log(\pi_A), \quad y=-log(\pi_G), \quad z=-log(\pi_C), \quad  t
=-log(\pi_T), \quad 
 a=-log(\alpha), \quad b=-log(\beta).
\end{eqnarray*}
Actually, this map induces a bijection between the Lie algebra $\LL_{5.6b}=\lrc{\B^{\id}_1,\B^{\id}_2,\B^{\zeta_2},\B^{\xi}_1,\B^{\xi}_2}$  and the set of (not necessarily stochastic) rate matrices of HKY model. The inverse is given by  
\begin{eqnarray*}
 Q=\sum_{i\neq j} q_{ij}L_{ij} \mapsto \sum_{i\neq j} e^{-q_{ij}}L_{ij}.
\end{eqnarray*}
However, these two models have different essential properties. For instance, while Model 5.6b is given by the \emph{linear} variety $\LL_{5.6b}$,  it can be seen that the set of rate matrices of HKY model describes a variety that is not linear and contains  singular points. 
The deep connection between Lie Markov models and submodels of the general time reversible model appears as a beautiful line of research that will deserve some attention from the authors in the future. 
%
\qed
\end{remark}

\begin{remark}
As already noted in Remark \ref{no_continuum}, a number of models in the above list have more symmetries than those requested by the group $\G$, and they already appeared in \cite{LMM} as Lie Markov models with $ \SG_4$ symmetry. 
   For those models, the decomposition into irreducible representations of $\G$ can be obtained from the decomposition into irreducible representations of $\SG_4$ by applying the branching rule of Table \ref{tab: branching_rule} (cf. Table 2 of \cite{LMM}). 
Since there are no subgroups between $\G$ and $\SG_4$, we can conclude that the rest of models listed here do not have further symmetries.\qed
\end{remark}

\begin{remark}\label{rmk:algebras}
A vector subspace $\LL$ in $\LL_{GM}$ is a \emph{matrix algebra} if its multiplicatively closed, that is, if the product $XY$ lies in $\LL$ for any couple  $X,Y\in \LL$. Of course, this condition is stronger than that of a Lie algebra. The reader may wonder which of the models in the above list are actually algebras. 
The authors were surprised to find that the only Lie algebras which are not algebras correspond to the Lie Markov models that appear in families depending on some parameters $a,b$ in the sense of Remark \ref{no_continuum}, that is, the Lie Markov models corresponding to the following decompositions:

\begin{description}
\item $\id$ (for dimension 1);  
 \item $\id \oplus \sgn$, $\id \oplus \zeta_1$ and $\id \oplus \zeta_2$ (for dimension 3); 
 \item $\id \oplus \zeta_{2} \oplus \xi$, $\id \oplus \sgn \oplus \xi$,  $\id \oplus \sgn \oplus \zeta_{1} \oplus \zeta_{2}$ (for dimension 4); and
 \item $\id \oplus \sgn \oplus \zeta_{1} \oplus \zeta_{2} \oplus \xi$ (for dimension 8).
\end{description}
Notice that these decompositions correspond exactly to those decomposition of Table~\ref{tab:all_decs} that arise as irreducible permutation representations $\langle \G/H\rangle_{\mathbb{C}}$ for a subgroup $H$ of $\G$. \qed
\end{remark}

\section{Discussion}\label{conc}
Following the ideas of \cite{LMM}, in this paper we have discussed Lie Markov models with purine/pyrimidine symmetry.
This symmetry was mathematically expressed by taking the group of nucleotide permutations
\[\G=\{e,(AG),(CT),(AG)(CT),(AC)(GT),(AT)(GC),(ACGT),(ATGC)\}.\]
Our main motivation is that this symmetry may be of special interest to the biologist who wishes to deal with models preserving the specific grouping of nucleotides into purines and pyrimidines.  
In Section~\ref{sec2} we recalled some of the basic definitions on Lie Markov models and the required tools arising from representation theory of groups. 
At the same time, we introduce a new concept which is the stochastic cone of a Lie Markov model, being the set of stochastic rate matrices of the Lie Markov model. 
In Section~\ref{sec3} we explained how to derive Lie Markov models with prescribed symmetry and discussed the geometry of the corresponding  cone of stochastic rate matrices.
In Section~\ref{sec4} we took the permutation group $\G$ and decomposed the space of all rate matrices into irreducible modules of $\G$ and provided a basis consistent with this decomposition. 
In Section~\ref{sec5} we gave the full list of all Lie Markov models with $\G$ symmetry, arranged by their dimension. 

In Section~\ref{sec2} we defined evolutionary models from a rate matrix perspective as some well-defined linear subspaces of the  space $\LL_{GM}$ of all rate matrices. 
We could focus on the substitution matrices instead, and keeping in mind the importance of substitution matrices being multiplicatively closed (see \cite{LMM}), define ``evolutionary model'' as some well-defined groups $\mathfrak{M}$ of matrices in $M_n(\mathbb{R})$. 
Then, when we restrict to the stochastic setting, we are led to consider the intersection of $\mathfrak{M}$ with the \emph{stochastic} polytope: 
\begin{eqnarray*}
 \mathbb{P}_{sto}=\left \{M=(m_{ij})\in \mathfrak{M} \mid  m_{ij}\geq 0, \sum_i m_{ij}=1 \right \}.
\end{eqnarray*}
The reader may note that $\mathbb{P}_{sto}$ is a compact polytope with the identity matrix in one of the vertices. This polytope is cut into several connected components by the algebraic hypersurface of equation $det(M)=0$. 
From a biological point of view, we are mainly interested in the connected component that contains the identity matrix. 
This is because, by continuity arguments, this connected component contains the exponential of the stochastic rate matrices of the model. 
In this paper, we have preferred to introduce evolutionary models from the point of view of rate matrices because both the definition of Lie Markov models and the procedure to construct them appear in a natural way in this setting. 
However, the connection between rate matrices and substitution models is not completely clear, and it deserves further attention. 
An interesting question on this issue is whether the image of the exponential map restricted to the stochastic cone covers the whole connected component of the identity. 
We want to explore this question in the future to clarify the connection between substitution and rate matrices. 


Although we have kept the original definition of symmetry for a Lie Markov model of \cite{LMM}, an interesting question arises if one tries to expand this definition. Namely, we could investigate Lie Markov models which are  invariant under the action of some permutation subgroup $G$ of $\SG_4$ without the additional request that they have a \emph{permutation basis}. Since the action of $\SG_4$ is linear, this would lead to strongly convex polyhedral cones. %
From an applied point of view, we do not find any particular reason not to consider this expanded definition, which would lead to a huge number of possible models. For example, we would admit the complex span of the ray-orbits $\trip{4}{\frac{1}{3}}{\frac{2}{3}}d: \LL =\langle R_{13}^+, R_{14}^+, R_{23}^+, R_{24}^+\rangle_{\CC} $, $\trip{4}{\frac{1}{3}}{\frac{2}{3}}e: \LL = \langle H_{13}^+, H_{14}^+, H_{23}^+, H_{24}^+\rangle_{\CC}$ and $\trip{4}{\frac{1}{3}}{\frac{2}{3}}f:  \LL = \langle V_{13}^+, V_{14}^+, V_{23}^+, V_{24}^+\rangle_{\CC}$ (see Table \ref{tab:Z2models} 
 as models with symmetry $\G$ and decomposition $\id \oplus \xi$. The reader may note that this decomposition does not appear in the list of Table \ref{tab:Z2models}. 
More interestingly, it is not hard to show that the set of doubly stochastic rate matrices has $\SG_4$ symmetry under this expanded definition, and moreover forms a Lie algebra.
The authors keep back this line of research for future publication.

\begin{table}[h]
\caption{\label{tab:Liebracket} \small{ The Lie brackets of the basis $\{\B^{*}_j\}$ of $\LL_{GM}$. The entries not included are easily determined by applying the rule $[X,Y]=-[Y,X]$. Here we use the notation $E_{3,5}=-6\B^{\id}_1+2\B^{\id}_2-2\B^{\zeta_2}_1$,
 $E_{3,6}=-2\B^{\sgn}-2 \B^{\zeta_1}$,
 $E_{4,5}=-2\B^{\sgn}+2\B^{\zeta_1}$,
 $E_{4,6}=-6\B^{\id}_1+2\B^{\id}_2+2\B^{\zeta_2}_1$.}}
\centering
\vspace*{2mm}
\begin{tabular}{c||cccccccccccc}
& $\B^{\id}_1$ & $\B^{\id}$ & $\B^{\sgn}$ &  $\B^{\zeta_1}$ & $\B^{\zeta_2}_1$ & $\B^{\zeta_2}$ & $\B^{\xi}_1$ & $\B^{\xi}_2$ & $\B^{\xi}_3$ & $\B^{\xi}_4$ & $\B^{\xi}_5$ & $\B^{\xi}_6$ \\ \hline \hline
$\B^{\id}_1$  & 0 & 0 &  0 & 0 & 0 & 0 & $-2\B^{\xi}_1$ & $-2\B^{\xi}_2$ & $-2\B^{\xi}_3$ & $-2\B^{\xi}_4$ & $2\B^{\xi}_5$ & $2\B^{\xi}
_6$ \\
$\B^{\id}$ &   & 0 & 0 & 0 & 0 & $-4\B^{\zeta_2}$ & $-4\B^{\xi}_1$ & $-4\B^{\xi}_2$ & 0 & 0 & 0 &  0 \\ 
$\B^{\sgn}$ &   &   & 0 & 4 $\B^{\zeta_2}_1$ & 4 $\B^{\zeta_1}$  & 0 &  $2\B^{\xi}_2$ & $2\B^{\xi}_1$ & $-2\B^{\xi}_4$ & $-2\B^{\xi}_3$ & $2\B^{\xi}_6$ & $2\B^{\xi}_5$ \\
$\B^{\zeta_1}$ &   &   &   & 0 & 4 $\B^{\sgn}$ & 0 & $-2\B^{\xi}_2$  & $2\B^{\xi}_1$ & $2\B^{\xi}_4$ & $-2\B^{\xi}_3$ & $2\B^{\xi}_6$ & $-2\B^{\xi}_5$  \\ 
$\B^{\zeta_2}_1$  &   &   &   &    & 0 & 0 & $-2\B^{\xi}_1$ & $2\B^{\xi}_2$& $-2\B^{\xi}_3$ & $2\B^{\xi}_4$ & $2\B^{\xi}_5$ & $-2\B^{\xi}_6$ \\
$\B^{\zeta_2}$ &   &   &   &   &   & 0 & 0 & 0 & 4 $\B^{\xi}_1$ & $-4\B^{\xi}_2$ & 0 & 0 \\ 
$\B^{\xi}_1$  &   &   &   &   &   &   & 0 & 0 & 0 & 0 & $2\B^{\zeta_2}$ & 0 \\
$\B ^{\xi}_2$  &   &   &   &   &   &   &   & 0 & 0 & 0 & 0 & $-2\B^{\zeta
_2}_2$\\
$\B^{\xi}_3$ &   &   &   &   &   &   &   &   & 0 & 0 & $E_{3,5}$ & $E_{3,6}$ \\
$\B^{\xi}_4$ &   &   &   &   &   &   &   &   &   & 0 & $E_{4,5}$  &
$E_{4,6}$ \\
$\B^{\xi}_5$ &   &   &   &   &   &   &   &   &   &   & 0 & 0\\
$\B^{\xi}_6$ &   &   &   &   &   &   &   &   &   &   &   & 0  \\
\end{tabular}
\end{table}



\normalsize

\quad


\begin{landscape}
\begin{table}[t]
{\renewcommand{\arraystretch}{1.2}}
\footnotesize
\centering
\begin{tabular}{|p{1.2cm}|p{12cm}|p{5.5cm}|}
\hline
orbit & matrices & {dec. as an abstract set / dec. in $\LL_{GM}$}\\
\hline \hline
$\trip{1}{1}{0}$ &  $\P_{(12)(34)}$ &$\id$\\ 
$\trip{1}{0}{1}$ & $\P_{(13)(24)}+\P_{(14)(23)}$ & $\id$  \\
$\trip{1}{\frac{1}{3}}{\frac{2}{3}}$ & $\P_{(12)(34)}+\P_{(13)(24)}+\P
_{(14)(23)}$ &$\id$  \\ \hline \hline
$\trip{2}{0}{1}a$ & $\P_{(13)(24)}, \P_{(14)(23)}$ & $\id \oplus \sgn$  \\
$\trip{2}{0}{1}b$ & $\P_{(1324)}, \P_{(1423)}$ & $\id \oplus \zeta_1$  \\ 
$\trip{2}{0}{1}c$ & $\Q{12}, \Q{34}$ & $\id \oplus \zeta_2$  \\
$\trip{2}{1}{0}$ & $\P_{(12)}, \P_{(34)}$ & $\id \oplus \zeta_2$  \\

$\trip{2}{\frac{1}{3}}{\frac{2}{3}}$ & $R_{12}^{+}, R_{34}^{+}$ & $\id \oplus \zeta_2$   \\\hline \hline
$\trip{4}{0}{1}a$ & $V_1+H_2-2\P_{12},V_2+H_1-2\P_{12}, V_3+H
_4-2\P_{34},V_4+H_3-2\P_{34}$ & $\id  \oplus \zeta_2 \oplus \xi$ / $\id  \oplus  \xi$ \\
$\trip{4}{0}{1}b$  & $H_{13}^{+}-(L_{21}+L_{43}), H_{14}^{+}-(L_{21}+L_{34}), H_{23}^{+}-(L_{12}+L_{43}), 
H_{24}^{+}-(L_{12}+L_{34})$ & $\id  \oplus \zeta_2 \oplus \xi$ / $\id  \oplus  \xi$  \\
$\trip{4}{0}{1}c$  &  $L_{13}+L_{14},L_{23}+L_{24},L_{31}+L_{32},L_{41}+L_{42}$ & $\id  \oplus \zeta_2 \oplus \xi$ \\
$\trip{4}{0}{1}d$ & $L_{13}+L_{42},L_{23}+L_{41},L_{31}+L_{24},L_{32}+L_{14} $ & $\id \oplus \sgn \oplus \xi$ \\
$\trip{4}{0}{1}e$  & $L_{13}+L_{24},L_{23}+L_{14},L_{31}+L_{42},L_{32}+L_{41} $ & $\id \oplus \sgn \oplus \zeta_1 \oplus \zeta_2$  \\
$\trip{4}{0}{1}f$ &  $\P_{(13)},\P_{(14)}, \P_{(23)}, \P_{(24)}$ & $\id  \oplus \zeta_2 \oplus \xi$ \\
$\trip{4}{1}{0}a$  & $L_{12},L_{21}, L_{34},L_{43}$ & $\id  \oplus \zeta_2 \oplus \xi$ / $\id  \oplus  \xi$   \\
$\trip{4}{1}{0}b$ & $L_{12}+L_{34},L_{12}+L_{43},L_{21}+L_{34},L_{21}+L_{43} $ & $\id  \oplus \zeta_2 \oplus \xi$ / $\id  \oplus  \xi$  \\ 
$\trip{4}{\frac{1}{3}}{\frac{2}{3}}a$ & $R_1,R_2,R_3,R_4$ &$\id  \oplus \zeta_2 \oplus \xi$ / $\id  \oplus  \xi$  \\
$\trip{4}{\frac{1}{3}}{\frac{2}{3}}b$ & $V_1,V_2,V_3,V_4$ & $\id  \oplus \zeta_2 \oplus \xi$ / $\id  \oplus  \xi$    \\
$\trip{4}{\frac{1}{3}}{\frac{2}{3}}c$ & $H_1,H_2,H_3,H_4$ & $\id  \oplus \zeta_2 \oplus \xi$ / $\id  \oplus  \xi$   \\
$\trip{4}{\frac{1}{3}}{\frac{2}{3}}d$ & $R_{13}^{+}, R_{14}^{+}, R_{23}^{+}, R_{24}^{+}$ & $\id \oplus \sgn \oplus \xi$ / $\id  \oplus  \xi$ \\
$\trip{4}{\frac{1}{3}}{\frac{2}{3}}e$ & $H_{13}^{+}, H_{14}^{+}, H_{23}^{+}, H_{24}^{+}$ &$\id \oplus \sgn \oplus \xi$ / $\id  \oplus  \xi$  \\
$\trip{4}{\frac{1}{3}}{\frac{2}{3}}f$ & $V_{13}^{+}, V_{14}^{+}, V_{23}^{+}, V_{24}^{+}$ & $\id \oplus \sgn \oplus \xi$ / $\id  \oplus  \xi$  \\
$\trip{4}{\frac{1}{7}}{\frac{6}{7}}$ & $V_1+\Q{12}, V_2+\Q{12}, V_3+\Q{34}, V_4+\Q{34}$ & $\id  \oplus \zeta_2 \oplus \xi$ / $\id  \oplus  \xi$   \\
$\trip{4}{\frac{1}{2}}{\frac{1}{2}}a$ & $R_{13}^{+}-(L_{14}+L_{32}), R_{14}^{+}-(L_{13}+L_{42}), R_{23}^{+}-(L_{24}+L_{31}), R_{24}^{+}-(L_{23}+L_{41})$ & $\id \oplus \zeta_2 \oplus \xi$ \\
$\trip{4}{\frac{1}{2}}{\frac{1}{2}}b$ & $\P_{(1234)},\P_{(1243)},\P_{(1324)},\P_{(1342)}$  &$\id \oplus \sgn \oplus \xi$ \\
$\trip{4}{\frac{1}{5}}{\frac{4}{5}}a$ & $2R_1+\Q{34}, 2R_2+\Q{34}, 2R_3+\Q{12}, 2R_4+\Q{12}$ & $\id \oplus \zeta_2 \oplus \xi$ \\
$\trip{4}{\frac{1}{5}}{\frac{4}{5}}b$ & $2H_1+\Q{34}, 2H_2+\Q{34}, 2H_3+\Q{12}, 2H_4+\Q{12}$ &$\id \oplus \zeta_2 \oplus \xi$\\
$\trip{4}{\frac{1}{5}}{\frac{4}{5}}c$ & $2V_1+\Q{12}, 2V_2+\Q{12}, 2V_3+\Q{34}, 2V_4+\Q{34}$ & $\id \oplus \zeta_2 \oplus \xi$\\
$\trip{4}{\frac{3}{5}}{\frac{2}{5}}$ & $V_1+\P_{(34)}, V_2+\P_{(34)}, V_3+\P_{(12)}, V_4+\P_{(12)}$ & $\id \oplus \zeta_2 \oplus \xi$\\ \hline \hline
$\trip{8}{0}{1}a$ & $ L_{13}, L_{14},L_{23},L_{24},L_{31},L_{32},L_{41},L_{42}$ & $\id \oplus \sgn \oplus \zeta_1 \oplus \zeta_2 \oplus 2 \xi$\\  
$\trip{8}{0}{1}b$ & $ \P_{(13)(24)}-L_{13}+L_{23},\P_{(13)(24)}-L_{24}+L_{14}, \P_{(13)(24)}-L_{31}+L_{41},\P_{(13)(24)}-L_{42}+L_{32}, 
\P_{(14)(23)}-L_{14}+L_{24}, \P_{(14)(23)}-L_{23}+L_{13}, \P_{(14)(23)}-L_{32}+L_{42},\P_{(14)(23)}-L_{41}+L_{31}$
& $\id \oplus \sgn \oplus \zeta_1 \oplus \zeta_2 \oplus 2 \xi$ / $\id \oplus \sgn \oplus \zeta_1 \oplus \xi$\\
$\trip{8}{\frac{1}{3}}{\frac{2}{3}}a$ & $R_1-L_{14}+L_{41}, R_1-L_{13}+L_{31}, R_2-L_{24}+L_{42}, R_2-L_{23}+L_{32}, 
R_3-L_{31}+L_{13}, R_2-L_{32}+L_{23}, R_4-L_{41}+L_{14}, R_4-L_{42}+L_{24}$ & $\id \oplus \sgn \oplus \zeta_1 \oplus \zeta_2 \oplus 2 \xi$/ $\id \oplus \zeta_1 \oplus \zeta_2 \oplus 2 \xi$ \\
$\trip{8}{\frac{1}{3}}{\frac{2}{3}}b$ & $\P_{(123)},\P_{(124)},\P_{(132)},\P_{(134)},\P_{(142)},\P_{(143)},\P_{(234)},\P_{(243)}$ & $\id \oplus \sgn \oplus \zeta_1 \oplus \zeta_2 \oplus 2 \xi$/ $\id \oplus \zeta_1 \oplus \zeta_2 \oplus 2 \xi$  \\
$\trip{8}{\frac{1}{2}}{\frac{1}{2}}$ &$L_{12}+L_{34}+2L_{13}, L_{12}+L_{34}+2L_{31}, L_{12}+L_{43}+2L_{14}, L_{12}+L_{43}+2L_{41}, 
L_{21}+L_{34}+2L_{23}, L_{21}+L_{34}+2L_{32}, L_{21}+L_{43}+2L_{24}, L_{21}+L_{43}+2L_{42}$ &$\id \oplus \sgn \oplus \zeta_1 \oplus \zeta_2 \oplus 2 \xi$/ $\id \oplus \zeta_1 \oplus \zeta_2 \oplus 2 \xi$ \\\hline
\end{tabular}
\caption{ \label{tab:ray_orbits} \small{Ray-orbits of the Lie Markov models with $\G$ symmetry with the corresponding generators. The 3rd column describes the decomposition into irreducible representations of both the abstract set generated by these orbits and the subspace of $\LL_{GM}$ spanned by them (see Remark \ref{ray_orbits}). When both decompositions are equal, we write it down only once. }}
\end{table}
\end{landscape}

\bibliographystyle{plain}
\bibliography{masterAB}

\begin{thebibliography}{10}

\bibitem{convex}
A.~D. Alexandrov.
\newblock {\em Convex polyhedra}.
\newblock Springer Monographs in Mathematics. Springer-Verlag, Berlin, 2005.
\newblock Translated from the 1950 Russian edition by N. S. Dairbekov, S. S.
  Kutateladze and A. B. Sossinsky, With comments and bibliography by V. A.
  Zalgaller and appendices by L. A. Shor and Yu. A. Volkov.

\bibitem{bogopolski}
O.~Bogopolski.
\newblock {\em Introduction to group theory}.
\newblock EMS Textbooks in Mathematics. European Mathematical Society (EMS),
  Z\"urich, 2008.
\newblock Translated, revised and expanded from the 2002 Russian original.

\bibitem{campbell1897}
J.~E. Campbell.
\newblock On a law of combination of operators (second paper).
\newblock {\em Proc. London Math. Soc.}, 28:381390, 1897.

\bibitem{casfer2010}
M.~Casanellas and J.~Fern\'andez-S\'anchez.
\newblock Relevant phylogenetic invariants of evolutionary models.
\newblock {\em J. Math. Pure. Appl.}, 96:207--229, 2010.

\bibitem{cassull}
M.~Casanellas and S.~Sullivant.
\newblock The strand symmetric model.
\newblock In {\em Algebraic statistics for computational biology}, pages
  305--321. Cambridge Univ. Press, New York, 2005.

\bibitem{draisma2008}
J.~Draisma and J.~Kuttler.
\newblock On the ideals of equivariant tree models.
\newblock {\em Math. Ann.}, 344:619--644, 2008.

\bibitem{felsenstein1981}
J.~Felsenstein.
\newblock Evolutionary trees from {DNA} sequences: a maximum likelihood
  approach.
\newblock {\em J. Mol. Evol.}, 17:368--376, 1981.

\bibitem{LMMweb}
J.~Fern{\'a}ndez-S{\'a}nchez.
\newblock Code for lie markov models with purine/pyrimidine symmetry.
\newblock
  \verb+http://www.pagines.ma1.upc.edu/~jfernandez/purine_pyrimidine.html+
  (2013).

\bibitem{HKY}
M.~Hasegawa, H.~Kishino, and T.~Yano.
\newblock Phylogenetic inference from {DNA} sequence data.
\newblock In {\em Statistical theory and data analysis, {II} ({T}okyo, 1986)},
  pages 1--13. North-Holland, Amsterdam, 1988.

\bibitem{james2001}
G.~James and M.~Liebeck.
\newblock {\em Representations and characters of groups}.
\newblock Cambridge University Press, New York, second edition, 2001.

\bibitem{johnson1985}
J.~E. Johnson.
\newblock {M}arkov-type {Lie} groups in {$GL(n,{R})$}.
\newblock {\em J. Math. Phys.}, 26:252--257, 1985.

\bibitem{JC69}
TH~Jukes and CR~Cantor.
\newblock Evolution of protein molecules.
\newblock {\em In Mammalian Protein Metabolism}, pages 21--132, 1969.

\bibitem{kimura1980}
M~Kimura.
\newblock A simple method for estimating evolutionary rates of base
  substitution through comparative studies of nucleotide sequences.
\newblock {\em J. Mol. Evol.}, 16:111--120, 1980.

\bibitem{kimura1981}
M.~Kimura.
\newblock Estimation of evolutionary distances between homologous nucleotide
  sequences.
\newblock {\em Proc. Natl. Acad. Sci.}, 78:1454--1458, 1981.

\bibitem{posada1998}
D.~Posada and K.~A. Crandall.
\newblock Modeltest: testing the model of {DNA} substitution.
\newblock {\em Bioinformatics}, 14:817--818, 1998.

\bibitem{rotman}
J.J. Rotman.
\newblock {\em An introduction to the theory of groups}, volume 148 of {\em
  Graduate Texts in Mathematics}.
\newblock Springer-Verlag, New York, fourth edition, 1995.

\bibitem{sagan2001}
B.~E. Sagan.
\newblock {\em The Symmetric Group: Representations, Combinatorial Algorithms,
  and Symmetric Functions. Second Edition.}
\newblock Graduate Texts in Mathematics. Springer, 2001.

\bibitem{semple2003}
C.~Semple and M.~Steel.
\newblock {\em Phylogenetics}.
\newblock Oxford Press, 2003.

\bibitem{SAGE}
W.\thinspace{}A. Stein et~al.
\newblock {\em {S}age {M}athematics {S}oftware ({V}ersion 4.8)}.
\newblock The Sage Development Team, 2012.
\newblock {\tt http://www.sagemath.org}.

\bibitem{LMM}
J.~G. Sumner, J.~Fern\'andez-S\'anchez, and P.~D. Jarvis.
\newblock Lie markov models.
\newblock {\em J. Theor. Biol.}, 298:16--31, 2012.

\bibitem{sumner2012}
J.~G. Sumner, P.~D. Jarvis, J.~Fern\'andez-S\'anchez, B.~T. Kaine, Michael~D.
  Woodhams, and B.~R Holland.
\newblock Is the general time-reversible model bad for molecular phylogenetics?
\newblock {\em Syst. Biol.}, page To appear: 10.1093/sysbio/sys042, 2012.

\bibitem{tavare1986}
S~Tavar\'e.
\newblock Some probabilistic and statistical problems in the analysis of dna
  sequences.
\newblock {\em Lectures on Mathematics in the Life Sciences (American
  Mathematical Society)}, 17:57--86, 1986.

\bibitem{yap_pachter}
V.B.Yap and L.~Pachter.
\newblock Identification of evolutionary hotspots in the rodent genomes.
\newblock {\em Genome Research}, 14(4):574--579, 2004.

\bibitem{woodhams2012}
M.~D. Woodhams, J.~Fern\'andez-S\'anchez, and J.~G. Sumner.
\newblock Implementation and performance of closed {RY} evolution models,.
\newblock 2012.
\newblock \textit{In preparation}.

\end{thebibliography}

\end{document}